\documentclass[11pt,reqno]{amsart}

\usepackage{amsmath}
\usepackage{amssymb}
\usepackage{amscd}
\usepackage{latexsym}
\usepackage{graphicx}
\usepackage{color}

\newcommand{\wh}{\widehat}

\newcommand{\ispa}[1]{\langle \,#1 \,\rangle }

\newcommand{\sgn}{\operatorname{sgn}\nolimits} 
 
\newcommand{\ol}{\overline}

\newcommand{\mb}{\mathbb}
\newcommand{\dbar}{d\hspace{-1pt}\Bar{}\,}

\newcommand{\ccal}{\mathcal{C}}

\newcommand{\fcal}{\mathcal{F}}

\newcommand{\mcal}{\mathcal{M}}

\newcommand{\mr}{\mathrm} 
\newcommand{\ep}{\epsilon}
\newcommand{\re}{{\rm Re}\,}
\newcommand{\im}{{\rm Im}\,}

\newcommand{\dsp}{\displaystyle}
\newcommand{\wlim}{\operatornamewithlimits{{\rm w}-lim}}

\newcommand{\map}{\mathrm{Map}\,}

\newtheorem{thm}{{\sc Theorem}}[section]
\newtheorem{cor}[thm]{{\sc Corollary}}
\newtheorem{lem}[thm]{{\sc Lemma}}
\newtheorem{prop}[thm]{{\sc Proposition}}
\theoremstyle{definition}

\newtheorem{rem}[thm]{{\sc Remark}}

\begin{document}

\title[Eigenfunction expansion for quantum walks]
{An eigenfunction expansion formula for\\ one-dimensional two-state quantum walks}
\author{Tatsuya Tate}
\address{Mathematical Institute, Graduate School of Sciences, Tohoku University, 
Aoba, Sendai 980-8578, Japan. }
\email{tatsuya.tate.c6@tohoku.ac.jp}
\thanks{The author is partially supported by JSPS KAKENHI Grant Number  18K03267, 17H06465.}
\date{\today}

\renewcommand{\thefootnote}{\fnsymbol{footnote}}
\renewcommand{\theequation}{\thesection.\arabic{equation}}
\renewcommand{\labelenumi}{{\rm (\arabic{enumi})}}
\renewcommand{\labelenumii}{{\rm (\alph{enumii})}}
\numberwithin{equation}{section}

\begin{abstract}
The purpose of this paper is to give a direct proof of an eigenfunction expansion formula 
for one-dimensional 2-state quantum walks, which is an analog of 
that for Sturm-Liouville operators due to Weyl, Stone, Titchmarsh and Kodaira. 
In the context of the theory of CMV matrix it had been already established by Gesztesy-Zinchenko. 
Our approach is restricted to the class of quantum walks mentioned above whereas 
it is direct and it gives some important properties of Green functions. 
The properties given here enable us to give a concrete formula 
for a positive-matrix-valued measure, which gives directly the spectral measure, 
in a simplest case of the so-called two-phase model. 
\end{abstract}

\maketitle
\section{Introduction}\label{INTRO}
The quantum walks are certain unitary operators, defined below, in a discrete setting, 
and they are sometimes regarded as a quantum counterpart of the classical random walks. 
The homogeneous two-state quantum walks (in one dimension with constant coin matrix) is well understood 
(see, for example, \cite{K}, \cite{GJS}, \cite{ST}), and recently the scattering-theoretical aspect, 
as a perturbation of homogeneous walks, are intensively investigated 
(see \cite{Mo}, \cite{MS}, \cite{RSA1}, \cite{RSA2}, \cite{MSSSS}). 
The Schr\"{o}dinger operators in one dimension are often called the Sturm-Liouville operators and they are well-studied. 
Thus it would be rather natural to understand resemblances between one-dimensional quantum walks and Sturm-Liouville operators. 
The purpose of the present paper is to give a proof of an eigenfunction expansion formula for 1-dimensional 
two-state quantum walks which is analogous to classical formulas of Weyl \cite{W1}, \cite{W2}, Stone \cite{S}, 
Titchmarsh \cite{T} and Kodaira \cite{Kd} for Sturm-Liouville operators. 
The theory of eigenfunction expansion for Sturm-Liouville operators are discussed, for example, 
in \cite{Ma}, \cite{RS2}, \cite{RS3}, \cite{KoM} and a short review can be found in \cite{OPUC2}. 
Probabilistic aspects of one-dimensional quantum walks are also intensively investigated. 
The notion of {\it transfer matrix} is introduced in \cite{KKK} to construct stationary 
measures from eigenfunctions for quantum walks, and it is suitable for our analysis.  
Then our basic idea in this paper is to use transfer matrix to develop a theory analogous to that for Sturm-Liouville operators. 

Before going to explain our setting-up, we should mention about the work by Gesztesy-Zinchenko \cite{GZ} 
on Weyl-Titchmarsh theory for CMV matrices with Verblunsky coefficients in the unit disk. 
The notion of CMV matrices has been introduced by Cantero-Moral-Vel\'{a}zquez \cite{CMV} 
and has further developed and deepened by Simon \cite{OPUC1}, \cite{OPUC2}. 
The one-dimensional two-state quantum walks are special CMV matrices and the theory of CMV matrices 
applied to this class of quantum walks in \cite{CGMV}, \cite{FO} and other works. 
Therefore, many of our results in this paper are essentially contained in \cite{GZ}. 
However, since our presentation and proofs are direct without using the theory of CMV matrices, 
and formulas are given in usual representation of unitary evolutions for quantum walks. 
Although our approach only works for the class of quantum walks mentioned above, 
the setting of our presentation could have advantageous aspect when the 
quantum walks are applied and used in areas different from pure mathematics 
such as information sciences or quantum physics. 
Furthermore, it seems that a property of the Green function, 
Theorem $\ref{RF1}$ below, is new, and it can be used to give a concrete formula, 
Theorem $\ref{TP1}$, to compute the positive-matrix-valued measure, 
which gives directly the spectral resolution,  
for certain special simplest case of the so-called two-phase model \cite{EKST}. 

Now, let us prepare notation to mention some of results in the paper. 
All the inner products in the paper are complex linear in the first variable and the anti-complex linear in the second. 
We denote the standard Hermitian inner product of two dimensional complex vector space $\mb{C}^{2}$ 
by $\ispa{\cdot,\cdot}_{\mb{C}^{2}}$ and the standard orthonormal basis of $\mb{C}^{2}$ by $\{e_{L}, e_{R}\}$, 
\[
e_{L}=
\begin{bmatrix}
1 \\
0
\end{bmatrix}
,\quad 
e_{R}=
\begin{bmatrix}
0 \\
1
\end{bmatrix}.  
\]
The orthogonal projection onto each 1-dimensional subspace, $\mb{C}e_{L}$, $\mb{C}e_{R}$, is denoted by $\pi_{L}$, $\pi_{R}$. 
In general, the set of maps from a set $X$ to another set $Y$ is denoted by $\map(X,Y)$. 
We fix $\mathcal{C} \in \map(\mb{Z},\mr{U}(2))$, where $\mr{U}(2)$ is the group of unitary $2 \times 2$ matrices, and define a linear map 
\begin{equation}\label{QW1}
U(\ccal) \colon \map(\mb{Z},\mb{C}^{2}) \to \map(\mb{Z},\mb{C}^{2})
\end{equation}
by the following formula, 
\begin{equation}\label{QW2}
U(\ccal)\Psi(x)=\pi_{L} \ccal(x+1) \Psi(x+1) +\pi_{R} \ccal(x-1) \Psi(x-1) 
\end{equation}
where $\Psi \in \map(\mb{Z},\mb{C}^{2})$, $x \in \mb{Z}$. 
Let $\ell^{2}(\mb{Z},\mb{C}^{2})$ be the Hilbert space of $\ell^{2}$-functions whose inner product is given by 
\[
\ispa{f,g}=\sum_{x \in \mb{Z}} \ispa{f(x),g(x)}_{\mb{C}^{2}}
\]
for $f,g \in \ell^{2}(\mb{Z},\mb{C}^{2})$. 
The linear map $U(\ccal)$ defined in $\eqref{QW2}$ becomes a unitary operator on $\ell^{2}(\mb{Z},\mb{C}^{2})$ 
when it is restricted to $\ell^{2}(\mb{Z},\mb{C}^{2})$, and it preserves the space $C_{0}(\mb{Z},\mb{C}^{2})$ 
of finitely supported $\mb{C}^{2}$-valued functions. 
In this paper, we call the linear map defined in $\eqref{QW1}$, $\eqref{QW2}$ the {\it quantum walk} with 
the {\it coin matrix} $\ccal \colon \mb{Z} \to \mr{U}(2)$. 
We write
\begin{equation}\label{coin1}
\ccal(x)=
\begin{bmatrix}
a_{x} & b_{x} \\
c_{x} & d_{x}
\end{bmatrix}
\quad 
\triangle_{x}=\det \ccal(x)=a_{x}d_{x}-b_{x}c_{x} \quad (x \in \mb{Z}). 
\end{equation}
Throughout this paper, we assume the following. 
\begin{equation}\label{Ass1}
a_{x} \neq 0 \qquad (\forall x \in \mb{Z}). 
\end{equation}
Under the assumption $\eqref{Ass1}$, the unitarity of the matrix $\ccal(x)$ causes 
$d_{x} \neq 0$ for any $x \in \mb{Z}$. 
\begin{thm}[\cite{KKK}, \cite{KS}, \cite{MSSSS}]\label{TM1}
Suppose that the coin matrix $\ccal$ satisfies the assumption $\eqref{Ass1}$. 
For any $x \in \mb{Z}$ and any $\lambda \in \mb{C}^{\times}:=\mb{C} \setminus \{0\}$, 
we define the $2 \times 2$ matrix $T_{\lambda}(x)$ by 
\begin{equation}\label{TRM1}
T_{\lambda}(x)=
\begin{bmatrix}
{\dsp \frac{1}{a_{x+1}} (\lambda -\lambda^{-1}c_{x} b_{x+1}) }& 
{\dsp -\lambda^{-1} \frac{b_{x+1}d_{x}}{a_{x+1}} }\\
{\dsp \lambda^{-1} c_{x} }&{\dsp \lambda^{-1} d_{x}}
\end{bmatrix}, 
\end{equation}
and the $2 \times 2$ matrix $F_{\lambda}(x)$ by 
\begin{equation}\label{TRM2}
F_{\lambda}(x)  = 
\begin{cases}
I & (x=0), \\
{\dsp T_{\lambda}(x-1) T_{\lambda}(x-2) \cdots T_{\lambda}(1) T_{\lambda}(0) }& (x \geq 1), \\
{\dsp T_{\lambda}(x)^{-1} T_{\lambda}(x+1)^{-1} \cdots T_{\lambda}(-2)^{-1} T_{\lambda}(-1)^{-1} }& (x \leq -1). 
\end{cases}
\end{equation}
For any $u \in \mb{C}^{2}$, we define $\Phi_{\lambda}(u) \in \map(\mb{Z},\mb{C}^{2})$ by 
\begin{equation}\label{EF1}
\Phi_{\lambda}(u)(x)  = F_{\lambda}(x)u \quad (x \in \mb{Z}). 
\end{equation}
Then the map $\Phi_{\lambda} \colon \mb{C}^{2} \to \map(\mb{Z},\mb{C}^{2})$ is injective 
and the eigenspace $\mcal^{\lambda}$ of $U(\ccal)$ in $\map(\mb{Z},\mb{C}^{2})$ with 
an eigenvalue $\lambda \in \mb{C}^{\times}$ coincides with the image of $\Phi_{\lambda}$. 
Hence $\dim \mcal^{\lambda}=2$ for each such $\lambda$. 
Furthermore, we have $\dim \mcal^{\lambda} \cap \ell^{2}(\mb{Z},\mb{C}^{2}) \leq 1$ for any $\lambda \in S^{1}$. 
\end{thm}
The matrices $T_{\lambda}(x)$ is called the {\it transfer matrix}\footnote{The $2 \times 2$ matrix-valued function 
$T_{\lambda}(x)$ is introduced in \cite{KKK}. There is another matrix called transfer matrix used in \cite{KS}. See Appendix in this paper.}, 
and it is useful to describe various functions and quantities related to the quantum walk $U(\ccal)$. 
For example, the Green function $R_{\lambda}(x,y)$, which is the kernel function of the resolvent 
\begin{equation}\label{RES1}
R(\lambda)=(U-\lambda)^{-1}\quad (\lambda \in \mb{C} \setminus \sigma(U)) 
\end{equation}
of the restriction $U$ of $U(\ccal)$ to $\ell^{2}(\mb{Z},\mb{C}^{2})$, where $\sigma(U)$ 
denotes the spectrum of the operator $U$, can be described as follows. 
\begin{thm}\label{Green}
We define $\mr{z}_{L},\mr{z}_{R} \in \map(\mb{Z},\mr{M}_{2}(\mb{C}))$, 
where $\mr{M}_{2}(\mb{C})$ denotes the space of complex $2 \times 2$ matrices, by 
\begin{equation}\label{ZLR1}
\mr{z}_{L}(x)=
\begin{bmatrix}
1 & 0 \\
-c_{x}/d_{x} & 0 
\end{bmatrix},\quad 
\mr{z}_{R}(x) = 
\begin{bmatrix}
0 & -b_{x} /a_{x} \\
0 & 1
\end{bmatrix}. 
\end{equation}
We set $\mr{x}_{0}(\lambda)=R_{\lambda}(0,0)$. Then, we have 
\begin{equation}\label{GF1}
R_{\lambda}(x,y)=
\begin{cases}
F_{\lambda}(x) [\mr{x}_{0}(\lambda)+\lambda^{-1} \mr{z}_{L}(0)] F_{\lambda^{*}} (y)^{*} 
& (y<x), \\
F_{\lambda}(x) [\mr{x}_{0}(\lambda)+\lambda^{-1} \mr{z}_{R}(0)] F_{\lambda^{*}} (y)^{*} 
& (y>x),
\end{cases}
\end{equation}
where $\lambda^{*}=\ol{\lambda}^{-1}$, and for each $x \in \mb{Z}$, we have 
\begin{equation}\label{GF2}
\begin{split}
R_{\lambda}(x,x) & = F_{\lambda}(x) [\mr{x}_{0}(\lambda) +\lambda^{-1}\mr{z}_{L}(0)]F_{\lambda^{*}} (x)^{*} 
-\lambda^{-1} \mr{z}_{L}(x) \\
& = F_{\lambda}(x) [\mr{x}_{0}(\lambda) +\lambda^{-1}\mr{z}_{R}(0)]F_{\lambda^{*}} (x)^{*} 
-\lambda^{-1} \mr{z}_{R}(x)
\end{split}
\end{equation}
\end{thm}
As above, the matrix-valued holomorphic function $\mr{x}_{0}(\lambda)=R_{\lambda}(0,0)$ 
plays one of central roles in the theory developed in the present paper. 
\begin{thm}\label{infinity}
Let $\lambda \in \mb{C}^{\times} \setminus S^{1}$. 
Let $A_{L}$ (resp. $A_{R}$) be the vector subspace in $\mb{C}^{2}$ consisting of 
all the vector $w \in \mb{C}^{2}$ satisfying 
\[
\sum_{x \geq N} \|F_{\lambda}(x)w\|_{\mb{C}^{2}}^{2}<+\infty 
\quad \left(
{\text resp. } \sum_{x \leq -N} \|F_{\lambda}(x)w\|_{\mb{C}^{2}}^{2}<+\infty 
\right) 
\]
for some $N>0$. Then we have $\dim A_{L}=\dim A_{R}=1$. In particular, we have 
\[
\mr{rank} [\mr{x}_{0}(\lambda) +\lambda^{-1}\mr{z}_{L}(0)]=
\mr{rank} [\mr{x}_{0}(\lambda) +\lambda^{-1}\mr{z}_{R}(0)]=\mr{rank}R_{\lambda}(x,y)=1
\]
for any $\lambda \in \mb{C}^{\times} \setminus S^{1}$ and $x,y \in \mb{Z}$ with $x \neq y$. 
\end{thm}
For $\lambda \in \mb{C}^{\times} \setminus S^{1}$, let $\mr{v}_{\pm}(\lambda)$ be unit vectors satisfying 
\begin{equation}\label{VPM1}
\sum_{x \geq 1} \|F_{\lambda}(x) \mr{v}_{+}(\lambda) \|_{\mb{C}^{2}}^{2}<+\infty,\quad 
\sum_{x \leq -1} \|F_{\lambda}(x) \mr{v}_{-}(\lambda) \|_{\mb{C}^{2}}^{2}<+\infty. 
\end{equation}
The existence of these unit vectors is assured by Theorem $\ref{infinity}$. For any unit vector $u=\,\!^{t}[a, b] \in \mb{C}^{2}$, 
we denote $u^{\perp}=\,\!^{t}[-\ol{b}, \ol{a}]$. 
\begin{thm}\label{RF1}
For $\lambda \in \mb{C}^{\times} \setminus S^{1}$ the unit vectors $\mr{v}_{+}(\lambda), \mr{v}_{-}(\lambda)$ 
are linearly independent. The matrix-valued holomorphic function $\mr{x}_{0}(\lambda)=R_{\lambda}(0,0)$ is given by 
\begin{equation}\label{VX000}
\begin{split}
\mr{x}_{0}(\lambda)e_{L} & =-\lambda^{-1} \frac{\ispa{\mr{z}_{L}(0)e_{L}, \mr{v}_{+}(\lambda)^{\perp}}_{\mb{C}^{2}}}
{\ispa{\mr{v}_{-}(\lambda), \mr{v}_{+}(\lambda)^{\perp}}_{\mb{C}^{2}}} \mr{v}_{-}(\lambda),\\
\mr{x}_{0}(\lambda)e_{R}& = -\lambda^{-1} \frac{\ispa{\mr{z}_{R}(0)e_{R}, \mr{v}_{-}(\lambda)^{\perp}}_{\mb{C}^{2}}}
{\ispa{\mr{v}_{+}(\lambda), \mr{v}_{-}(\lambda)^{\perp}}_{\mb{C}^{2}}}\mr{v}_{+}(\lambda). 
\end{split}
\end{equation}
\end{thm}
The eigenfunction expansion theorem due to Weyl \cite{W1}, \cite{W2}, Stone \cite{S}, Titchmarsh \cite{T} and Kodaira \cite{Kd} 
(where historical notes can be found in \cite{Kd} and \cite{RS2}) 
is regarded as an inversion formula for a generalized Fourier transform defined by 
eigenfunctions for Sturm-Liouville operators. 
Let us state eigenfunction expansion formula for the quantum walk $U(\ccal)$ defined 
by the coin matrix $\ccal$ satisfying $\eqref{Ass1}$. 
For any $f \in C_{0}(\mb{Z},\mb{C}^{2})$, we define a function $\fcal_{\ccal}[f]$ on $\mb{C}^{\times}$ by 
\begin{equation}\label{QWF1}
\fcal_{\ccal}[f](\lambda)=\wh{f}^{\ccal}(\lambda):=\sum_{x \in \mb{Z}} F_{\lambda^{*}}(x)^{*}f(x),\qquad 
\lambda \in \mb{C}^{\times},\ \lambda^{*}:=1/\ol{\lambda}. 
\end{equation}
The sum in $\eqref{QWF1}$ is finite because the support of $f$ is finite. 
Therefore, the function $\fcal_{\ccal}[f]$ is a Laurent polynomial in $\lambda \in \mb{C}^{\times}$. 
We call $\fcal_{\ccal}[f]=\wh{f}^{\ccal}$ the {\it QW-Fourier transform} of $f$. 
\begin{thm}\label{WTK}
There exists a positive-matrix-valued measure $\Sigma$ on $S^{1}$ such that we have the following. 
\par
\medskip
\noindent{\rm (1)} The resolvent $R(\lambda)$ is written as 
\begin{equation}\label{GF3}
\ispa{R(\lambda)f,g}=\int_{S^{1}} \frac{1}{\zeta -\lambda} 
\ispa{d\Sigma(\zeta)\fcal_{\ccal}[f](\zeta),\fcal_{\ccal}[g](\zeta)}_{\mb{C}^{2}}. 
\end{equation}
The positive-matrix valued measure $\Sigma$ satisfying $\eqref{GF3}$ is unique. 
\par
\medskip
\noindent{\rm (2)} For any $f,g \in C_{0}(\mb{Z},\mb{C}^{2})$, we have 
\begin{equation}\label{PF1}
\ispa{f,g}=\int_{S^{1}} \ispa{d\Sigma(\zeta)\fcal_{\ccal}[f](\zeta),\fcal_{\ccal}[g](\zeta)}_{\mb{C}^{2}}. 
\end{equation}
\noindent{\rm (3)} Let 
\[
U=\int_{S^{1}} \lambda \,dE(\lambda)
\]
be the spectral resolution of the unitary operator $U$, where $E$ is a projection-valued measure on $S^{1}$. 
Then, for each Borel set $A$ in $S^{1}$, the projection $E(A)$ on $\ell^{2}(\mb{Z},\mb{C}^{2})$ 
is written as 
\begin{equation}\label{Sk1}
[E(A)f](x)=\int_{A} F_{\zeta}(x) d\Sigma(\zeta) \fcal_{\ccal}[f](\zeta),\quad f \in C_{0}(\mb{Z},\mb{C}^{2}). 
\end{equation}
In particular the following inversion formula holds for $f \in C_{0}(\mb{Z},\mb{C}^{2})$. 
\begin{equation}\label{INV1}
f(x)=\int_{S^{1}} F_{\zeta}(x) d\Sigma(\zeta) \fcal_{\ccal}[f](\zeta). 
\end{equation}
\end{thm}
\begin{cor}\label{spect}
The following holds. 
\begin{enumerate}
\item The spectrum $\sigma(U)$ coincides with the support of $\Sigma$; 
\item $\lambda \in S^{1}$ is an eigenvalue of $U$ if and only if $\Sigma(\{\lambda\}) \neq 0$. 
When $\lambda$ is an eigenvalue of $U$, the projection $E(\{\lambda\})$  onto the eigenspace of $\lambda$ 
is given by 
\[
[E(\{\lambda\})f](x)=F_{\lambda}(x)\Sigma(\{\lambda\}) \fcal_{\ccal}[f](\lambda). 
\]
\end{enumerate}
\end{cor}
We refer the readers to \cite{GT}, \cite{GZ} for properties of matrix-valued measures. 
We note that the matrix-valued function $\mr{x}_{0}(\lambda)$ is not an m-Carath\'{e}odory function 
in the sense of \cite{OPUC1} because our operator is unitary. 
Instead, we use the matrix 
\begin{equation}\label{xmat}
\mr{x}(\lambda)=I+2\lambda \mr{x}_{0}(\lambda)
\end{equation}
which is indeed a m-Carath\'{e}odory function.  
The positive-matrix-valued measure $\Sigma$ is then 
a boundary value of the function $\mr{x}(\lambda)$ in the sense that the $\Sigma$ satisfies 
\begin{equation}\label{xmat2}
\mr{x}(\lambda)=\int_{S^{1}} \frac{\zeta +\lambda}{\zeta -\lambda}\,d\Sigma(\zeta)\quad (|\lambda| <1),  
\end{equation}
and the positive-matrix-valued measure $\Sigma$ is characterized as the boundary value of $\mr{x}_{0}(\lambda)$, 
\begin{equation}\label{xmat3}
d\Sigma(\zeta)=\wlim_{r \uparrow 1} \re \mr{x}(r\zeta)=
\wlim_{r \uparrow 1}[r \zeta \mr{x}_{0}(r\zeta) -r^{-1}\zeta \mr{x}_{0}(r^{-1}\zeta)].
\end{equation}
Let $C(S^{1}\!,\mb{C}^{2})$ be the space of continuous $\mb{C}^{2}$-valued function on $S^{1}$. 
For any $k, l \in C(S^{1},\mb{C}^{2})$, we define 
\begin{equation}\label{FS1}
\ispa{k,l}_{\Sigma}=\int_{S^{1}} \ispa{d\Sigma(\zeta)k(\zeta),l(\zeta)}_{\mb{C}^{2}}. 
\end{equation}
It is clear that this defines an inner product on $C(S^{1}\!,\mb{C}^{2})$. 
We denote $L^{2}(S^{1}\!,\mb{C}^{2})_{\Sigma}$ the completion of $C(S^{1}\!,\mb{C}^{2})$ by the norm 
defined by the inner product $\eqref{FS1}$. 
\begin{thm}\label{FTT}
The map $\fcal_{\ccal}$ extends to a unitary operator from 
$\ell^{2}(\mb{Z},\mb{C}^{2})$ to $L^{2}(S^{1}\!, \mb{C}^{2})_{\Sigma}$. 
The quantum walk $U$ on $\ell^{2}(\mb{Z},\mb{C}^{2})$ is unitarily equivalent to the operator 
defined by the multiplication by $\lambda \in S^{1}$ on $L^{2}(S^{1}\!, \mb{C}^{2})_{\Sigma}$, namely, we have 
\begin{equation}\label{QWF2}
\fcal_{\ccal}[U(\ccal)f](\lambda)=\lambda \fcal_{\ccal}[f](\lambda)
\end{equation}
for any $f \in \ell^{2}(\mb{Z},\mb{C}^{2})$. 
\end{thm}
The organization of the paper is as follows. In Section $\ref{CONJU}$ we solve two equations, 
an inhomogeneous eigenvalue equation and its conjugate. 
The definition of QW-Fourier transform $\eqref{QWF1}$ comes from the fact 
that it gives a defect of the left-inverse obtained by solving conjugate equation to 
the inhomogeneous eigenvalue equation to be the right-inverse. 
See Theorem $\ref{RESW}$ for a precise statement. 
The solutions to these equations are used to prove Theorem $\ref{WTK}$ in Section $\ref{GREENT}$. 
Some of properties of the Green kernel, such as Theorem $\ref{infinity}$, is proved also in Section $\ref{GREENT}$. 
In Section $\ref{INTEGRAL}$ we give proofs of Theorems $\ref{WTK}$ and $\ref{FTT}$. 
Finally we compute two example, homogeneous quantum walks and a simplest case of the two-phase model in Section $\ref{EXS}$. 
\section{Inhomogeneous eigenvalue equations and its conjugate}\label{CONJU}
Let $f \in \map(\mb{Z},\mb{C}^{2})$, $\lambda \in \mb{C}^{\times}$ and $w \in \mb{C}^{2}$. 
We consider the following equation, 
\begin{equation}\label{IEE1}
(U(\ccal)-\lambda)\Psi=f,\quad \Psi(0)=w. 
\end{equation}
Any $\Psi(a)$ with $a \in \mb{Z}$ can be chosen for initial value, 
but we have chosen $a=0$ for simplicity of notation.
To prove Theorem $\ref{Green}$, it is important to construct solutions to the equation $\eqref{IEE1}$ 
and its conjugate equation
\begin{equation}\label{IEE2}
(U(\ccal)^{*} -\ol{\lambda})\Phi=f,\quad \Phi(0)=w, 
\end{equation}
where the map 
\begin{equation}\label{CQW1}
U(\ccal)^{*} \colon \map(\mb{Z},\mb{C}^{2}) \to \map(\mb{Z},\mb{C}^{2})
\end{equation} 
is the extension of the formal adjoint (on $C_{0}(\mb{Z},\mb{C}^{2})$) of $U(\ccal)$ given by 
\begin{equation}\label{CQW2}
(U(\ccal)^{*}\Phi )(x)=\ccal(x)^{*}\pi_{L} \Phi(x-1)+\ccal(x)^{*}\pi_{R} \Phi(x+1) \quad (x \in \mb{Z}). 
\end{equation}
The formal adjoint $A^{*}$ of a linear map $A \colon \map(\mb{Z},\mb{C}^{2}) \to \map(\mb{Z},\mb{C}^{2})$ 
preserving $C_{0}(\mb{Z},\mb{C}^{2})$ is the unique linear map $A^{*} \colon C_{0}(\mb{Z},\mb{C}^{2})$ satisfying 
\[
\ispa{A^{*}f,g}=\ispa{f,Ag} \quad (f,g \in C_{0}(\mb{Z},\mb{C}^{2})). 
\]
\begin{prop}\label{SOLP}
\noindent{\rm (1)} \hspace{1pt} 
We define a kernel function $v_{\lambda} \in \map(\mb{Z}^{2},\mr{M}_{2}(\mb{C}))$ as follows. 
\begin{equation}\label{Vker}
\begin{split}
v_{\lambda}(x,0) & =
\begin{cases}
\lambda^{-1} F_{\lambda}(x) \mr{z}_{L}(0) & (x \geq 1), \\
\lambda^{-1}F_{\lambda}(x) \mr{z}_{R}(0) & (x \leq -1),
\end{cases}
\\
v_{\lambda}(x,y) & = 
\begin{cases}
\lambda^{-1}F_{\lambda}(x)F_{\lambda}(y)^{-1}(\mr{z}_{L}(y)-\mr{z}_{R}(y)) & (1 \leq y \leq x-1), \\
-\lambda^{-1} \mr{z}_{R}(x) & (1 \leq y =x),  \\
\lambda^{-1} F_{\lambda}(x) F_{\lambda}(y)^{-1} (\mr{z}_{R}(y) -\mr{z}_{L}(y)) & (x+1 \leq y \leq -1), \\
-\lambda^{-1} \mr{z}_{L}(x) & (x=y \leq -1), \\
0 & (\text{{\rm otherwise}}).
\end{cases}
\end{split}
\end{equation}
Then for each $f \in \map(\mb{Z},\mb{C}^{2})$, $\lambda \in \mb{C}^{\times}$ and $w \in \mb{C}^{2}$, 
the equation $\eqref{IEE1}$ has a unique solution given by 
\begin{equation}\label{SOL1}
\Psi=\Phi_{\lambda}(w)+V_{\lambda}f, 
\end{equation}
where $\Phi_{\lambda}(w)$ is defined in $\eqref{EF1}$ and $V_{\lambda}f$ is defined by  
\begin{equation}\label{VLAM}
V_{\lambda}f(x)=
\sum_{y \in \mb{Z}} v_{\lambda}(x,y)f(y). 
\end{equation}
\noindent{\rm (2)} We define a kernel function $w_{\lambda}^{o} \in \map(\mb{Z}^{2},\mr{M}_{2}(\mb{C}^{2}))$ 
as follows. 
\begin{equation}\label{W0ker}
\begin{split}
w_{\lambda}^{o}(x,0) & = 
\begin{cases}
\lambda^{*}F_{\lambda^{*}}(x)\mr{z}_{R}(0)^{*} & (x \geq 1), \\
\lambda^{*}F_{\lambda^{*}}(x)\mr{z}_{L}(0)^{*} & (x \leq -1),
\end{cases}
\\
w_{\lambda}^{o}(x,y) & = 
\begin{cases}
\lambda^{*}F_{\lambda^{*}}(x)F_{\lambda^{*}}(y)^{-1}(\mr{z}_{R}(y)^{*} -\mr{z}_{L}(y)^{*})  & 
(1 \leq y \leq x-1), \\
-\lambda^{*}\mr{z}_{L}(x)^{*} & 
(1 \leq y =x), \\
\lambda^{*} F_{\lambda^{*}}(x)F_{\lambda^{*}}(y)^{-1} (\mr{z}_{L}(y)^{*}-\mr{z}_{R}(y)^{*})   & 
(x+1 \leq y \leq -1), \\
-\lambda^{*}\mr{z}_{R}(x)^{*} & (x=y \leq -1), \\
0 & (\text{{\rm otherwise}}). 
\end{cases}
\end{split}
\end{equation}
Then for each $f \in \map(\mb{Z},\mb{C}^{2})$, $\lambda \in \mb{C}^{\times}$ and $w \in \mb{C}^{2}$, 
the equation $\eqref{IEE2}$ has a unique solution given by 
\begin{equation}\label{SOL2}
\Phi=\Phi_{\lambda^{*}}(w)+W_{\lambda}^{o}f, 
\end{equation}
where $W_{\lambda}^{o}f$ is defined by 
\begin{equation}\label{WLAM}
W_{\lambda}^{o}f(x)=
\sum_{y \in \mb{Z}} w_{\lambda}^{o}(x,y)f(y). 
\end{equation}
\end{prop}
We need some of the following formulas to prove Proposition $\ref{SOLP}$. 
\begin{lem}\label{FMS}
For any $\lambda \in \mb{C}^{\times}$, $x \in \mb{Z}$, we have the following. 
\begin{enumerate}
\item $\dsp F_{\lambda}(x+1)=T_{\lambda}(x)F_{\lambda}(x)$, 
\item $\dsp \pi_{L}\ccal(x) T_{\lambda}(x-1)=\lambda \pi_{L}$, 
\item $\dsp \pi_{R} \ccal(x) T_{\lambda}(x)^{-1} =\lambda \pi_{R}$, 
\item $\dsp \mr{z}_{L}(x)=\frac{\lambda}{a_{x+1}} T_{\lambda}(x)^{-1}\pi_{L}=\frac{\triangle_{x}}{d_{x}}\ccal(x)^{*}\pi_{L}$, 
\item $\dsp \mr{z}_{R}(x)=\frac{\lambda}{d_{x-1}} T_{\lambda}(x-1) \pi_{R}=\frac{\triangle_{x}}{a_{x}} \ccal(x)^{*} \pi_{R}$, 
\item $\dsp \mr{z}_{L}(x)^{*}+\mr{z}_{R}(x)=\mr{z}_{L}(x)+\mr{z}_{R}(x)^{*}=I$, 
\item $a_{x}\mr{z}_{L}(x)^{*}+d_{x}\mr{z}_{R}(x)^{*}=\pi_{L} \ccal(x) \mr{z}_{L}(x)^{*}+\pi_{R}\ccal(x)\mr{z}_{R}(x)^{*}=\ccal(x)$, 
\item $T_{\lambda}(x)[\mr{z}_{L}(x)-\mr{z}_{R}(x)]T_{\lambda^{*}}(x)^{*}=[\mr{z}_{L}(x+1)-\mr{z}_{R}(x+1)]$. 
\end{enumerate}
\end{lem}
These formulas can be checked directly from the definitions. 
\par
%\vspace{10pt}
\medskip
\noindent {\it Proof of Proposition $\ref{SOLP}$.}\hspace{3pt}
Since $v_{\lambda}(x,y)$ and $w_{\lambda}^{o}(x,y)$ do not vanish only when 
$y$ lies between $0$ and $x$, the sums in $\eqref{VLAM}$ and $\eqref{WLAM}$ are finite. 
Thus, $V_{\lambda}f$ and $W_{\lambda}^{o}f$ are well-defined as elements in $\map(\mb{Z},\mb{C}^{2})$. 
The proof of two assertions (1), (2) are similar and hence we only give a proof of (2). 
Since $U(\ccal)U(\ccal)^{*}=I$, the uniqueness of the solution to $\eqref{IEE2}$ follows from Theorem $\ref{TM1}$. 
Thus it is enough to check that the function defined by $\eqref{SOL2}$ solves the equation $\eqref{IEE2}$. 
Since the function $\Phi_{\lambda^{*}}(w)$ satisfies $(U(\ccal)^{*}-\ol{\lambda})\Phi_{\lambda^{*}}(w)=0$, 
it is enough to check that the function $\Phi=W_{\lambda}^{o}f$ solves the equation $\eqref{IEE2}$. 
For simplicity of notation we set 
\[
A_{\lambda}(x,y):=\ccal(x)^{*}\pi_{L} w_{\lambda}^{o}(x-1,y) +
\ccal(x)^{*}\pi_{R} w_{\lambda}^{o}(x+1,y) -\ol{\lambda} w_{\lambda}^{o}(x,y). 
\]
Then it is enough to show
\begin{equation}\label{SOLw}
A_{\lambda}(x,y)=\delta_{x,y}I 
\end{equation}
for any $x,y \in \mb{Z}$, where $I$ is $2 \times 2$ identity matrix and $\delta_{x,y}$ is Kronecker's delta. 
Suppose first that $x \geq 2$. 
By the definition of $w_{\lambda}^{o}$ and Lemma $\ref{FMS}$, (1), (2), (3), we have 
\[
\begin{split}
A_{\lambda}(x,0) & =
\lambda^{*}\left\{
\ccal(x)^{*}\pi_{L} F_{\lambda^{*}}(x-1) +
\ccal(x)^{*}\pi_{R} F_{\lambda^{*}}(x+1) 
\right\}\mr{z}_{R}(0)^{*} - F_{\lambda^{*}}(x)\mr{z}_{R}(0)^{*} \\
& = \lambda^{*}\ccal(x)^{*}\left\{
\pi_{L}T_{\lambda}(x-1)^{-1} +\pi_{R} T_{\lambda^{*}}(x) 
\right\}F_{\lambda^{*}}(x) \mr{z}_{R}(0)^{*}- F_{\lambda^{*}}(x)\mr{z}_{R}(0)^{*} \\
& =\lambda^{*} \ccal(x)^{*}\left\{
\lambda^{*-1} \pi_{L} \ccal(x) +\lambda^{*-1} \pi_{R} \ccal(x)
\right\}F_{\lambda^{*}}(x) \mr{z}_{R}(0)^{*}- F_{\lambda^{*}}(x)\mr{z}_{R}(0)^{*} \\
& = \ccal(x)^{*}\ccal(x)F_{\lambda^{*}}(x) \mr{z}_{R}(0)^{*}- F_{\lambda^{*}}(x)\mr{z}_{R}(0)^{*}=0. 
\end{split}
\]
This computation also holds for $x=1$ because $F_{\lambda^{*}}(0)=I$ and $\pi_{L} \mr{z}_{R}(0)^{*}=0$. 
For $1 \leq y \leq x-2$, similar computation shows 
\[
\begin{split}
A_{\lambda}(x,y) & = \lambda^{*} \ccal(x)^{*} 
\left\{
\pi_{L} T_{\lambda^{*}}(x-1)^{-1} +\pi_{R} T_{\lambda^{*}} (x) 
\right\} F_{\lambda^{*}} (x) F_{\lambda^{*}} (y)^{-1} 
[\mr{z}_{R}(y)^{*} -\mr{z}_{L}(y)^{*}] \\
& \hspace{20pt} - F_{\lambda^{*}} (x)F_{\lambda^{*}} (y)^{-1} [\mr{z}_{R}(y)^{*} -\mr{z}_{L}(y)^{*}] \\
& = \lambda^{*} \ccal(x)^{*} \left\{
\lambda^{*-1} \pi_{L}\ccal(x) +\lambda^{*-1} \pi_{R} \ccal(x) 
\right\} F_{\lambda^{*}} (x) F_{\lambda^{*}} (y)^{-1} 
[\mr{z}_{R}(y)^{*} -\mr{z}_{L}(y)^{*}] \\
& \hspace{20pt} - F_{\lambda^{*}} (x)F_{\lambda^{*}} (y)^{-1} [\mr{z}_{R}(y)^{*} -\mr{z}_{L}(y)^{*}] \\
& = 0
\end{split}
\]
When $y=x-1$, we see 
\[
\begin{split}
A_{\lambda}(x,x-1) & = -\lambda^{*} \ccal(x)^{*} \left\{
\pi_{L} T_{\lambda^{*}} (x-1)^{-1} +\pi_{R} T_{\lambda^{*}} (x) 
\right\} F_{\lambda^{*}} (x)F_{\lambda^{*}} (x-1)^{-1} \mr{z}_{L}(x-1)^{*} \\
& \hspace{20pt} +\lambda^{*} \ccal(x)^{*} \pi_{R} F_{\lambda^{*}}(x+1)F_{\lambda^{*}}(x-1)^{-1} \mr{z}_{R}(x-1)^{*} \\
& \hspace{20pt} -F_{\lambda^{*}} (x)F_{\lambda^{*}}(x-1)^{-1} [\mr{z}_{R}(x-1)^{*}-\mr{z}_{L}(x-1)^{*}] \\
& = \lambda^{*} \ccal(x)^{*} \pi_{R} F_{\lambda^{*}}(x+1) F_{\lambda^{*}}(x-1)^{-1} \mr{z}_{R}(x-1)^{*} 
-F_{\lambda^{*}}(x) F_{\lambda^{*}}(x-1)^{-1} \mr{z}_{R}(x-1)^{*} \\
& = \left(
\lambda^{*} \ccal(x)^{*} \pi_{R} T_{\lambda^{*}}(x) -I
\right) F_{\lambda^{*}}(x) F_{\lambda^{*}}(x-1)^{-1} \mr{z}_{R}(x-1)^{*} \\
& = -\ccal(x)^{*} \pi_{L} \ccal(x) T_{\lambda^{*}} (x-1) \mr{z}_{R}(x-1)^{*} \\
& = -\lambda^{*}\ccal(x)^{*} \pi_{L} \mr{z}_{R}(x-1)^{*}=0.
\end{split}
\]
A similar computation also works well for $A_{\lambda}(1,0)$. 
For $x \geq 1$, similar computation shows 
\[
A_{\lambda}(x,x)=\ccal(x)^{*} 
\left\{
\pi_{R} \ccal(x) \mr{z}_{R}(x)^{*} +\pi_{L} \ccal(x) \mr{z}_{L}(x)^{*}
\right\}=\ccal(x)^{*}\ccal(x)=I, 
\]
which shows the equation $\eqref{SOLw}$ for $x \geq 1$. 
Let us consider the case $x=0$. Since $w_{\lambda}^{o}(0,y)=0$, we see 
\[
A_{\lambda}(0,y)=\ccal(0)^{*} \pi_{L}w_{\lambda}^{o}(-1,y)+\ccal(0)^{*}\pi_{R} w_{\lambda}^{o}(1,y). 
\]
Since $\pi_{R}\mr{z}_{L}(y)^{*}=\pi_{L} \mr{z}_{R}(y)^{*}=0$ for any $y \in \mb{Z}$, we see $A_{\lambda}(0,\pm 1)=0$. 
By Lemma $\ref{FMS}$, (2), (3), (7), we have 
\[
\begin{split}
A_{\lambda}(0,0) & =\ccal(0)^{*} \left\{
\lambda^{*} \pi_{L} T_{\lambda^{*}}(-1)^{-1}\mr{z}_{L}(0)^{*} +\lambda^{*} \pi_{R} T_{\lambda^{*}}(0) \mr{z}_{R}(0)^{*}
\right\} \\
& = \ccal(0)^{*} \left\{
\pi_{L} \ccal(0) \mr{z}_{L}(0)^{*} +\pi_{R} \ccal(0) \mr{z}_{R}(0)^{*}
\right\} \\
& = \ccal(0)^{*} \ccal(0)=I, 
\end{split}
\]
which completes the proof of $\eqref{SOLw}$ for $x \geq 0$. 
The equation $\eqref{SOLw}$ for $x \leq -1$ can be checked by a similar computation and we omit it. 
\hfill$\square$
\par
%\vspace{10pt}
\medskip
\begin{rem}
The equation $\eqref{IEE2}$ is equivalent to the recurrence equation 
\begin{equation}\label{IEEX}
\begin{split}
\Phi(x+1) = & T_{\lambda^{*}}(x) \Phi(x) -\frac{\lambda^{*}}{a_{x+1}} \pi_{L} \ccal(x+1) f(x+1)  \\
& +\left(
I-\frac{1}{a_{x+1}} \pi_{L} \ccal(x+1) 
\right) \pi_{R} \ccal(x) f(x)
\end{split}
\end{equation}
with the condition $\Phi(0)=w$. The explicit formula for the kernel function $w_{\lambda}^{o}$ 
can be deduced from this recurrence equation. 
\hfill$\square$
\end{rem}
\medskip
\begin{cor}\label{WLAM2}
For $\lambda \in \mb{C}^{\times}$ we define $w_{\lambda} \in \map(\mb{Z}^{2},\mr{M}_{2}(\mb{C}))$ by 
\begin{equation}\label{Wker}
w_{\lambda}(x,y)=w_{\lambda}^{o}(y,x)^{*} \quad (x,y \in \mb{Z}). 
\end{equation}
Let $W_{\lambda} \colon C_{0}(\mb{Z},\mb{C}^{2}) \to \map(\mb{Z},\mb{C}^{2})$ be a 
map defined by 
\begin{equation}\label{OPW}
(W_{\lambda}f)(x)=\sum_{y \in \mb{Z}} w_{\lambda}(x,y)f(y)\quad (f \in C_{0}(\mb{Z},\mb{C}^{2})). 
\end{equation}
Then, $W_{\lambda}f$ satisfies 
\begin{equation}\label{CQW3}
W_{\lambda}(U(\ccal)-\lambda)f=f
\end{equation}
for any $f \in C_{0}(\mb{Z},\mb{C}^{2})$. 
\end{cor}
Note that $W_{\lambda}$ defined in $\eqref{OPW}$ is the formal adjoint of $W_{\lambda}^{o}$. 
Corollary $\ref{WLAM2}$ follows from the definition of the formal adjoint and Proposition $\ref{SOLP}$, (2). 
The kernel function $w_{\lambda}$ is given explicitly by the following. 
\begin{equation}\label{Wker2}
\begin{split}
w_{\lambda}(0,y) & =
\begin{cases}
\lambda^{-1} \mr{z}_{R}(0) F_{\lambda^{*}}(y)^{*} & (y \geq 1) \\
\lambda^{-1} \mr{z}_{L}(0) F_{\lambda^{*}}(y)^{*} & (y \leq -1)
\end{cases}
\\
w_{\lambda}(x,y) & = 
\begin{cases}
\lambda^{-1} (\mr{z}_{R}(x)-\mr{z}_{L}(x)) F_{\lambda^{*}}(x)^{*-1} F_{\lambda^{*}}(y)^{*} 
& (1 \leq x \leq y-1) \\
-\lambda^{-1} \mr{z}_{L}(x) & (1 \leq x=y) \\
\lambda^{-1} (\mr{z}_{L}(x)-\mr{z}_{R}(x)) F_{\lambda^{*}}(x)^{*-1} F_{\lambda^{*}}(y)^{*} 
& (y+1 \leq x \leq -1) \\
-\lambda^{-1} \mr{z}_{R} (x) & (y =x \leq -1) \\
0 & (\text{{\rm otherwise}})
\end{cases}
\end{split}
\end{equation}
One of the most important properties of the operator $W_{\lambda}$ is the following. 
\begin{thm}\label{RESW}
For any $f \in C_{0}(\mb{Z},\mb{C}^{2})$, we have $W_{\lambda}f \in C_{0}(\mb{Z},\mb{C}^{2})$ and 
\begin{equation}\label{RESWe}
(U -\lambda) W_{\lambda}f =f -\delta_{0} \otimes \wh{f}^{\ccal}(\lambda), 
\end{equation}
where the QW-Fourier transform $\wh{f}^{\ccal}$ of $f$ is defined in $\eqref{QWF1}$. 
\end{thm}
\begin{proof}
For fixed $y \in \mb{Z}$ the function $w_{\lambda}(x,y)$ in $x$ can not vanish 
only when $x$ lies between $0$ and $y$, and hence $W_{\lambda}f \in C_{0}(\mb{Z},\mb{C}^{2})$ for $f \in C_{0}(\mb{Z},\mb{C}^{2})$. 
Let $f \in C_{0}(\mb{Z},\mb{C}^{2})$. For $x \geq 2$, $(U(\ccal)-\lambda)W_{\lambda}f(x)$ is given as 
\[
\begin{split}
U(\ccal) W_{\lambda} f (x) & = \pi_{L} \ccal(x+1) W_{\lambda} f(x+1) + \pi_{R} \ccal(x-1) W_{\lambda} f(x-1) \\
& = \sum_{y \geq x+1} \pi_{L}\ccal(x+1) w_{\lambda}(x+1,y) f(y) 
+ \sum_{y \geq x-1} \pi_{R} \ccal(x-1) w_{\lambda}(x-1,y) f(y) 
\end{split} 
\]
By Lemma $\ref{FMS}$, (4) and $\eqref{Wker2}$, we have 
\[
\pi_{L} \ccal(x+1)w_{\lambda}(x+1,x+1)=-\lambda^{-1} \frac{\triangle_{x+1}}{d_{x+1}} \pi_{L},\quad 
\pi_{R}\ccal(x-1)w_{\lambda}(x-1,x-1)=0. 
\]
Similarly, we see 
\[
\begin{split}
\pi_{R} \ccal(x-1) w_{\lambda}(x-1,y) & =\lambda^{-1} \frac{\triangle_{x-1}}{a_{x-1}} 
\pi_{R} F_{\lambda^{*}}(x-1)^{*-1} F_{\lambda^{*}}(y)^{*} \quad (y \geq x), \\
\pi_{L} \ccal(x+1) w_{\lambda}(x+1,y) & = -\lambda^{-1} \frac{\triangle_{x+1}}{d_{x+1}} 
\pi_{L} F_{\lambda^{*}}(x+1)^{*-1} F_{\lambda^{*}}(y)^{*} \quad (y \geq x+2). 
\end{split}
\]
From this we have for $x \geq 2$ 
\[
\begin{split}
U(\ccal) W_{\lambda}f (x) = & \lambda^{-1} \frac{\triangle_{x-1}}{a_{x-1}} \pi_{R} T_{\lambda^{*}} (x-1)^{*} f(x) \\
 + & \lambda^{-1} \left(
\frac{\triangle_{x-1}}{a_{x-1}} \pi_{R} T_{\lambda^{*}} (x-1)^{*} -
\frac{\triangle_{x+1}}{d_{x+1}} \pi_{L} T_{\lambda^{*}}(x)^{*-1}
\right) F_{\lambda^{*}}(x)^{*-1} \sum_{y \geq x+1} F_{\lambda^{*}}(y)^{*} f(y). 
\end{split}
\]
By taking the conjugate of (4) and (5) in Lemma \ref{FMS} 
and using (6) in the same lemma, we see
\[
\begin{gathered}
\lambda^{-1}\frac{\triangle_{x+1}}{d_{x+1}} \pi_{L} T_{\lambda^{*}}(x)^{*-1} =\frac{1}{a_{x}} \pi_{L} \ccal(x)=\mr{z}_{L}(x)^{*},\\
\lambda^{-1}\frac{\triangle_{x-1}}{a_{x-1}} \pi_{R} T_{\lambda^{*}}(x-1)^{*}=\frac{1}{d_{x}} \pi_{R} \ccal(x)=\mr{z}_{R}(x)^{*},  
\end{gathered}
\]
and $\mr{z}_{R}(x)^{*}-\mr{z}_{L}(x)^{*}=\mr{z}_{R}(x)-\mr{z}_{L}(x)$. These formulas show
\[
U(\ccal)W_{\lambda}f (x)=f(x) + \lambda W_{\lambda}f (x)  
\]
for $x \geq 2$. For other cases $x \leq -2$ and $x =\pm 1$, the assertion follows from similar computation. 
Finally, we consider the case $x=0$. By Lemma $\ref{FMS}$ again, we see 
\[
\begin{split}
\pi_{L}\ccal(1) w_{\lambda}(1,y) & = -z_{L}(0)^{*} F_{\lambda^{*}}(y)^{*} =(-I+\mr{z}_{R}(0))F_{\lambda^{*}}(y)^{*} 
=-F_{\lambda^{*}}(y)^{*}+\lambda w_{\lambda}(0,y) \quad (y \geq 1),\\
\pi_{R}\ccal(-1) w_{\lambda}(-1,y) & = -\mr{z}_{R}(0)^{*} F_{\lambda^{*}}(y)^{*} =(-I+\mr{z}_{L}(0)) F_{\lambda^{*}}(y)^{*}
=-F_{\lambda^{*}}(y)^{*}+\lambda w_{\lambda}(0,y) \quad (y \leq -1). 
\end{split}
\]
From this we conclude 
\[
\begin{split}
U(\ccal)W_{\lambda}f(0) & = \sum_{y \geq 1} \pi_{L} \ccal(1) w_{\lambda}(1,y) f(y) 
+ \sum_{y \leq -1} \pi_{R} \ccal(-1) w_{\lambda}(-1,y) f(y)  \\
& = -\sum_{y \neq 0} F_{\lambda^{*}}(y)^{*}f(y) +\sum_{y \neq 0} \lambda w_{\lambda}(0,y)f(y) \\
& = \lambda W_{\lambda}f(0) -\wh{f}^{\ccal}(\lambda) +f(0), 
\end{split}
\]
which shows the assertion. 
\end{proof}
Before proceeding the proof of Theorem $\ref{Green}$, we prove the equation $\eqref{QWF2}$ for 
$f \in C_{0}(\mb{Z},\mb{C}^{2})$ in Theorem $\ref{FTT}$. 
\par
\medskip
\noindent{\it Proof of $\eqref{QWF2}$.}\hspace{3pt}
We first note that $T_{\lambda}(y)^{-1}$ is given by 
\[
T_{\lambda}(y)^{-1}=
\begin{bmatrix}
{\dsp \lambda^{-1}a_{y+1} }&{\dsp \lambda^{-1}b_{y+1} }\\[7pt]
{\dsp -\lambda^{-1} \frac{c_{y}a_{y+1}}{d_{y}} }&{\dsp \frac{1}{d_{y}}(\lambda -\lambda^{-1}b_{y+1}c_{y})} 
\end{bmatrix}. 
\]
By using this it can be shown directly that 
\[
T_{\lambda^{*}}(y-1)^{*-1} \pi_{L} \ccal(y)=\lambda \ccal(y)^{*}\pi_{L} \ccal(y),\quad 
T_{\lambda^{*}}(y)^{*} \pi_{R} \ccal(y)=\lambda \ccal(y)^{*}\pi_{R} \ccal(y). 
\]
From this formulas, and (1) in Lemma $\ref{FMS}$, we have for $f \in C_{0}(\mb{Z},\mb{C}^{2})$, 
\[
\begin{split}
\fcal_{\ccal}[U(\ccal)f](\lambda) & = \sum_{y \in \mb{Z}} F_{\lambda^{*}}(y)(U(\ccal)f)(y) \\
& = \sum_{y \in \mb{Z}} F_{\lambda^{*}}(y)^{*} 
\left(
\pi_{L}\ccal(y+1)f(y+1) +\pi_{R}\ccal(y-1) f(y-1) 
\right) \\
& = \sum_{y \in \mb{Z}} F_{\lambda^{*}}(y)^{*} T_{\lambda^{*}}(y-1)^{*-1} \pi_{L}\ccal(y)f(y) 
+\sum_{y \in \mb{Z}} F_{\lambda^{*}}(y)^{*}T_{\lambda^{*}}(y)^{*} \pi_{R} \ccal(y)f(y) \\
& = \lambda \sum_{y \in \mb{Z}} F_{\lambda^{*}}(y)^{*} 
\left(
\ccal(y)^{*} \pi_{L} \ccal(y)+ \ccal(y)^{*} \pi_{R} \ccal(y)
\right) f(y) \\
& = \lambda \sum_{y \in \mb{Z}} F_{\lambda^{*}}(y)^{*}f(y)=\lambda \fcal_{\ccal}[f],  
\end{split}
\]
which shows $\eqref{QWF2}$. 
\hfill$\square$
\par
\medskip
\section{Green function and its properties} \label{GREENT}
As in the previous sections, we denote $U$ the restriction of $U(\ccal)$ to 
the Hilbert space $\ell^{2}(\mb{Z},\mb{C}^{2})$. The operator $U$ is unitary 
and hence its spectrum, denoted $\sigma(U)$, is a subset of a unit circle $S^{1}$. 
For $\lambda \in \mb{C} \setminus \sigma(U)$, we denote $R(\lambda)$ the resolvent $(U-\lambda)^{-1}$. 
The Green function $R_{\lambda} \in \map(\mb{Z}^{2},\mr{M}_{2}(\mb{C}))$ 
is then defined by the formula 
\begin{equation}\label{DGREEN}
R_{\lambda}(x,y)u=R(\lambda)(\delta_{y} \otimes u)(x)
\end{equation}
for $\lambda \in \mb{C} \setminus \sigma(U)$, $x,y \in \mb{Z}$, and $u \in \mb{C}^{2}$. 
For $\lambda \in \mb{C} \setminus \sigma(U)$, we set 
\begin{equation}\label{X0}
\mr{x}_{0}(\lambda)=R_{\lambda}(0,0). 
\end{equation}
The function $\mr{x}_{0} \colon \mb{C} \setminus \sigma(U) \to \mr{M}_{2}(\mb{C})$ is holomorphic, 
and since $R(0)=U^{*}$, we have $\mr{x}_{0}(0)=0$ by $\eqref{CQW2}$. 
For any $f \in C_{0}(\mb{Z},\mb{C}^{2})$, the difference $R(\lambda)f-V_{\lambda}f \in \map(\mb{Z},\mb{C}^{2})$, 
where $V_{\lambda}f$ is defined in $\eqref{VLAM}$, satisfies
\[
(U(\ccal)-\lambda) (R(\lambda)f-V_{\lambda}f) =0. 
\]
Therefore, by Theorem $\ref{TM1}$, there exists a unique vector $\mr{x}_{\lambda}(f) \in \mb{C}^{2}$ such that 
\begin{equation}\label{DX1}
R(\lambda)f(x) =V_{\lambda}f(x) +\Phi_{\lambda}(\mr{x}_{\lambda}(f))(x) \quad (x \in \mb{Z}).  
\end{equation}
Since $V_{\lambda}f(0)=0$, we have 
\[
\mr{x}_{\lambda}(f)=R(\lambda)f(0),   
\]
and $\mr{x}_{0}(\lambda)u=\mr{x}_{\lambda}(\delta_{0} \otimes u)$. 
To prove Theorem $\ref{Green}$, we need the following. 
\begin{lem}\label{REC1}
For any $x \in \mb{Z}$ and $\lambda \in \mb{C}^{\times}$ we have 
\[
F_{\lambda}(x)[\mr{z}_{L}(0)-\mr{z}_{R}(0)]F_{\lambda^{*}}(x)^{*}=[\mr{z}_{L}(x)-\mr{z}_{R}(x)]. 
\]
\end{lem}
\begin{proof}
This follows from (8) in Lemma $\ref{FMS}$ and the definition of $F_{\lambda}(x)$ in $\eqref{TRM2}$. 
\end{proof}
\noindent{\it Proof of Theorem $\ref{Green}$.}\hspace{3pt} 
Setting $f=\delta_{y} \otimes u$ with $y \in \mb{Z}$, $u \in \mb{C}^{2}$ in $\eqref{DX1}$ gives 
\begin{equation}\label{GR1}
\begin{split}
R_{\lambda}(x,y)u & =v_{\lambda}(x,y)u+F_{\lambda}(x)\mr{x}_{\lambda}(\delta_{y} \otimes u) \\
& = v_{\lambda}(x,y)u+F_{\lambda}(x)R_{\lambda}(0,y)u. 
\end{split}
\end{equation}
Applying $R(\lambda)$ to the equation $\eqref{RESWe}$ in Theorem $\ref{RESW}$, we have for $f \in C_{0}(\mb{Z},\mb{C}^{2})$, 
\[
R(\lambda)f(x)=W_{\lambda}f(x) +R(\lambda) (\delta_{0} \otimes \wh{f}^{\ccal}(\lambda))
=W_{\lambda}f(x) +R_{\lambda}(x,0)\wh{f}^{\ccal}(\lambda).
\]
We set $f=\delta_{y} \otimes u$ in the above. Noting $\fcal_{\ccal}[\delta_{y} \otimes u](\lambda)=F_{\lambda^{*}}(y)^{*}u$, 
we have 
\begin{equation}\label{GR2}
R_{\lambda}(x,y)u=w_{\lambda}(x,y)u +R_{\lambda}(x,0)F_{\lambda^{*}}(y)^{*}u.
\end{equation}
Setting $y=0$ in $\eqref{GR1}$ gives 
\[
R_{\lambda}(x,0)=v_{\lambda}(x,0) +F_{\lambda}(x) \mr{x}_{0}(\lambda), 
\]
and substituting this into $\eqref{GR2}$ yields 
\begin{equation}\label{GR3}
\begin{split}
R_{\lambda}(x,y) & = w_{\lambda}(x,y) +v_{\lambda}(x,0)F_{\lambda^{*}}(y)^{*} +F_{\lambda}(x) \mr{x}_{0}(\lambda) F_{\lambda^{*}}(y)^{*} \\
& = 
\begin{cases}
{\dsp w_{\lambda}(x,y)+F_{\lambda}(x)[\mr{x}_{0}(\lambda) +\lambda^{-1}\mr{z}_{L}(0)] F_{\lambda^{*}}(y)^{*}}& (x \geq 1) \\
{\dsp w_{\lambda}(x,y)+F_{\lambda}(x)[\mr{x}_{0}(\lambda) +\lambda^{-1}\mr{z}_{R}(0)] F_{\lambda^{*}}(y)^{*}}& (x \leq -1) \\
{\dsp w_{\lambda}(0,y)+\mr{x}_{0}(\lambda)F_{\lambda^{*}}(y)^{*}}& (x=0) 
\end{cases}
\end{split}
\end{equation}
For $x=0$, it follows from $\eqref{GR3}$ that
\begin{equation}\label{GG1}
R_{\lambda}(0,y) =
\begin{cases}
{\dsp [\mr{x}_{0}(\lambda)+\lambda^{-1}\mr{z}_{R}(0)]F_{\lambda^{*}}(y)^{*}}& (y \geq 1) \\
{\dsp [\mr{x}_{0}(\lambda)+\lambda^{-1}\mr{z}_{L}(0)]F_{\lambda^{*}}(y)^{*}}& (y \leq -1). 
\end{cases}
\end{equation}
From $\eqref{Wker2}$, we have $w_{\lambda}(x,y)=0$ for the two cases $(x \leq -1,\, y \geq x+1)$ and $(x \geq 1,\, y \leq x-1)$, 
and hence we have 
\begin{equation}\label{GG2}
R_{\lambda}(x,y) = 
\begin{cases}
{\dsp F_{\lambda}(x)[\mr{x}_{0}(\lambda) +\lambda^{-1}\mr{z}_{R}(0)] F_{\lambda^{*}}(y)^{*}}& (x \leq -1,\, y \geq x+1) \\
{\dsp F_{\lambda}(x)[\mr{x}_{0}(\lambda) +\lambda^{-1}\mr{z}_{L}(0)] F_{\lambda^{*}}(y)^{*}}& (x \geq 1,\, y \leq x-1) 
\end{cases}
\end{equation}
To consider the other cases, we note that, by Lemma $\ref{REC1}$, 
\begin{equation}\label{Wker3}
w_{\lambda}(x,y) = 
\begin{cases}
{\dsp \lambda^{-1} F_{\lambda}(x)[\mr{z}_{R}(0)-\mr{z}_{L}(0)]F_{\lambda^{*}}(y)^{*} }& (1 \leq x \leq y-1) \\
{\dsp \lambda^{-1} F_{\lambda}(x)[\mr{z}_{L}(0)-\mr{z}_{R}(0)]F_{\lambda^{*}}(y)^{*} }& (y+1 \leq x \leq -1)
\end{cases}
\end{equation}
Therefore, $\eqref{GR3}$ gives 
\begin{equation}\label{GG3}
R_{\lambda}(x,y) = 
\begin{cases}
{\dsp F_{\lambda}(x)[\mr{x}_{0}(\lambda) +\lambda^{-1}\mr{z}_{R}(0)] F_{\lambda^{*}}(y)^{*}}& (1 \leq x \leq y-1) \\
{\dsp F_{\lambda}(x)[\mr{x}_{0}(\lambda) +\lambda^{-1}\mr{z}_{L}(0)] F_{\lambda^{*}}(y)^{*}}& (y+1 \leq x \leq -1) 
\end{cases}
\end{equation}
From $\eqref{GG1}$, $\eqref{GG2}$ and $\eqref{GG3}$ we conclude $\eqref{GF1}$. 
For $x=y \geq 1$, $\eqref{Wker2}$ and Lemma $\ref{REC1}$ shows 
\[
\begin{split}
R_{\lambda}(x,x) & =-\lambda^{-1}\mr{z}_{L}(x)+F_{\lambda}(x)[\mr{x}_{0}(\lambda)+\lambda^{-1}\mr{z}_{L}(0)]F_{\lambda^{*}}(x)^{*} \\
& = -\lambda^{-1} \mr{z}_{L}(x)+F_{\lambda}(x)[\mr{x}_{0}(\lambda)+
\lambda^{-1}\mr{z}_{R}(0)+\lambda^{-1}(\mr{z}_{L}(0)-\mr{z}_{R}(0))]F_{\lambda^{*}}(x)^{*} \\
& = -\lambda^{-1} \mr{z}_{L}(x)+\lambda^{-1}(\mr{z}_{L}(x)-\mr{z}_{R}(x)) +
F_{\lambda}(x)[\mr{x}_{0}(\lambda)+\lambda^{-1}\mr{z}_{R}(0)]F_{\lambda^{*}}(x)^{*} \\
& = -\lambda^{-1} \mr{z}_{R}(x)+
F_{\lambda}(x)[\mr{x}_{0}(\lambda)+\lambda^{-1}\mr{z}_{R}(0)]F_{\lambda^{*}}(x)^{*}, 
\end{split}
\]
which shows $\eqref{GF2}$ for $x =y\geq 1$. The same computation works well also for $x=y \leq -1$,  
and this completes the proof. 
\hfill$\square$
\par
\medskip
The Green function $R_{\lambda}(x,y)$ is, as above, 
expressed in terms of the products $F_{\lambda}(x)$ of the transfer matrices $T_{\lambda}(x)$ and 
the special value $\mr{x}_{0}(\lambda)=R_{\lambda}(0,0)$. 
Therefore, we will face the computation of the matrix-valued function $\mr{x}_{0}(\lambda)$ 
when we apply results in this paper. 
Theorem $\ref{infinity}$ and its Theorem $\ref{RF1}$ gives us one of methods 
to compute $\mr{x}_{0}(\lambda)$ whose proof is given in the rest of this section. 
The following phenomenon is one of most important in the proof of Theorem $\ref{infinity}$. 
\begin{lem}\label{HSL}
Let $\lambda \in \mb{C}^{\times} \setminus S^{1}$. Then we have 
\begin{equation}\label{HS1}
\sum_{x \geq 1} \|F_{\lambda}(x)\|_{\mr{HS}}^{2}=+\infty, \quad 
\sum_{x \leq -1} \|F_{\lambda}(x)\|_{\mr{HS}}^{2}=+\infty.  
\end{equation}
\end{lem}
The Hilbert-Schmidt norm of the $2 \times 2$ matrix 
\[
A=\begin{bmatrix}
a & b \\
c & d
\end{bmatrix}
\]
is given as $\|A\|_{\mr{HS}}^{2}=|a|^{2}+|b|^{2}+|c|^{2}+|d|^{2}$. Lemma $\ref{HSL}$ could be proved 
by a method in \cite{MSSSS}. See Appendix in the paper. But, we give here another proof. 
\begin{proof}
Let $\lambda \in \mb{C}^{\times} \setminus S^{1}$. We have the following formula.  
\begin{equation}\label{EFM}
\ccal(x) F_{\lambda}(x)=\lambda \pi_{L}F_{\lambda}(x-1) + \lambda \pi_{R}F_{\lambda}(x+1). 
\end{equation}
The formula $\eqref{EFM}$ is used to prove Theorem $\ref{TM1}$ and is proved by using (1), (2) in Lemma $\ref{FMS}$. 
Since the proof of assertions for two sums in $\eqref{HS1}$ is identical, we prove that the sum 
\[
D:=\sum_{x \geq 1} \|F_{\lambda}(x)\|_{\mr{HS}}^{2}
\]
is infinity. To prove this we assume $D<+\infty$ and deduce a contradiction. 
In general, if $C \in \mr{U}(2)$ then for any $A \in \mr{M}_{2}(\mb{C})$, we have 
\begin{equation}\label{HSA1}
\|CA\|_{\mr{HS}}=\|CA\|_{\mr{HS}},\quad 
\|A\|_{\mr{HS}}^{2}=\|\pi_{L}A\|_{\mr{HS}}^{2}+\|\pi_{R}\|_{\mr{HS}}^{2}. 
\end{equation}
The equation $\eqref{EFM}$ and the first of the identities $\eqref{HSA1}$ give 
\[
\|F_{\lambda}(x)\|_{\mr{HS}}^{2}=\|\ccal(x) F_{\lambda}(x)\|_{\mr{HS}}^{2} 
=|\lambda|^{2} \|\pi_{L} F_{\lambda}(x-1)\|_{\mr{HS}}^{2} +|\lambda|^{2} \|\pi_{R} F_{\lambda}(x+1)\|_{\mr{HS}}^{2}. 
\]
Therefore, we have 
\[
\begin{split}
D & = |\lambda|^{2} \sum_{x \geq 1} \|\pi_{L} F_{\lambda}(x-1)\|_{\mr{HS}}^{2} 
+\|\lambda|^{2} \sum_{x \geq 1}\|\pi_{R} F_{\lambda}(x+1)\|_{\mr{HS}}^{2} \\
& = |\lambda|^{2} \sum_{x \geq 0} \|\pi_{L} F_{\lambda}(x)\|_{\mr{HS}}^{2} 
+|\lambda|^{2} \sum_{x \geq 2} \|\pi_{R} F_{\lambda}(x)\|_{\mr{HS}}^{2} \\
& = |\lambda|^{2} \|\pi_{L} F_{\lambda}(0)\|_{\mr{HS}}^{2} -|\lambda|^{2} \|\pi_{R} F_{\lambda}(1)\|_{\mr{HS}}^{2} 
+|\lambda|^{2} \sum_{x \geq 1} \left(
\|\pi_{L} F_{\lambda}(x)\|_{\mr{HS}}^{2} +\|\pi_{R} F_{\lambda}(x)\|_{\mr{HS}}^{2}
\right) \\
& = |\lambda|^{2} \|\pi_{L} F_{\lambda}(0)\|_{\mr{HS}}^{2} -|\lambda|^{2} \|\pi_{R} F_{\lambda}(1)\|_{\mr{HS}}^{2} 
+|\lambda|^{2}D.
\end{split}
\]
By the facts $F_{\lambda}(0)=I$, $F_{\lambda}(1)=T_{\lambda}(0)$ and the definition of $T_{\lambda}(x)$ show 
\[
\|\pi_{L} F_{\lambda}(0) \|_{\mr{HS}}^{2} =1,\quad 
\|\pi_{R} F_{\lambda}(1) \|_{\mr{HS}}^{2} =|\lambda|^{-2} (|c_{0}|^{2} +|d_{0}|^{2})=|\lambda|^{-2}. 
\]
By substituting this into the above equation and by noting $|\lambda| \neq 1$ by the assumption 
we conclude  
\[
D=|\lambda|^{2} -1 +|\lambda|^{2}D, \quad \text{hence} \quad  D=-1, 
\]
which is a contradiction. 
\end{proof}
By Lemma $\ref{HSL}$ the subspace $A_{L}$ and $A_{R}$ in Theorem $\ref{infinity}$ 
has dimension less than or equal to $1$. The function $R(\lambda)(\delta_{0} \otimes u)$ is 
in $\ell^{2}(\mb{Z},\mb{C}^{2})$ for any $u \in \mb{C}^{2}$ and we have, by Theorem $\ref{Green}$, 
\[
\begin{split}
\|R(\lambda)(\delta_{0} \otimes u)\|^{2} & =\sum_{x \in \mb{Z}} \|R(\lambda)(\delta_{0} \otimes u)(x)\|_{\mb{C}^{2}}^{2} 
= \sum_{x \in \mb{Z}} \|R_{\lambda}(x,0)u\|_{\mb{C}^{2}}^{2} \\
& = \|\mr{x}_{0}(\lambda)u\|_{\mb{C}^{2}}^{2} 
+ \sum_{x \geq 1} \|F_{\lambda}(x) [\mr{x}_{0}(\lambda) +\lambda^{-1} \mr{z}_{L}(0) ] u\|_{\mb{C}^{2}}^{2} \\
& \hspace{20pt} 
+\sum_{x \leq -1} \|F_{\lambda}(x) [\mr{x}_{0}(\lambda) +\lambda^{-1} \mr{z}_{R}(0) ] u\|_{\mb{C}^{2}}^{2} <+\infty.
\end{split}
\]
Since $\dim A_{L} \leq 1$ and $\dim A_{R} \leq 1$, we see 
\[
\mr{rank}[\mr{x}_{0}(\lambda) +\lambda^{-1} \mr{z}_{L}(0) ] \leq 1,\quad 
\mr{rank}[\mr{x}_{0}(\lambda) +\lambda^{-1} \mr{z}_{R}(0) ] \leq 1. 
\]
Therefore, it is enough to prove the following lemma to conclude Theorem $\ref{infinity}$. 
\begin{lem}\label{RNK1}
We have 
\[
\mr{rank}[\mr{x}_{0}(\lambda) +\lambda^{-1} \mr{z}_{L}(0) ]=\mr{rank}[\mr{x}_{0}(\lambda) +\lambda^{-1} \mr{z}_{R}(0) ]=1
\]
\end{lem}
\begin{proof}
Suppose $\mr{rank}[\mr{x}_{0}(\lambda) +\lambda^{-1} \mr{z}_{L}(0) ]=0$, which means 
$\mr{x}_{0}(\lambda)=-\lambda^{-1} \mr{z}_{L}(0)$. Then for any $u \in \mb{C}^{2}$ we see 
\[
\sum_{x \leq -1} \|F_{\lambda}(x) [\mr{z}_{L}(0) -\mr{z}_{R}(0)]u\|_{\mb{C}^{2}}^{2}<+\infty. 
\]
However the matrix $\mr{z}_{L}(0) -\mr{z}_{R}(0)$ is non-singular and it contradicts Lemma $\ref{HSL}$. 
\end{proof}
\medskip
\noindent{\it Proof of Theorem $\ref{RF1}$.}\hspace{3pt}
Let $u \in \mb{C}^{2}$. Then $w \in \mb{C}^{2}$ satisfies 
\[
\sum_{x \geq 1} \|F_{\lambda}(x) [w +\lambda^{-1} \mr{z}_{L}(0)u] \|_{\mb{C}^{2}}^{2}<+\infty, 
\quad \text{{\rm and}} \quad 
\sum_{x \leq -1} \|F_{\lambda}(x) [w +\lambda^{-1} \mr{z}_{R}(0)u] \|_{\mb{C}^{2}}^{2}<+\infty
\]
if and only if $w=\mr{x}_{0}(\lambda)u$ since $\mr{x}_{0}(\lambda)u$ is the unique vector such that 
\[
R(\lambda)(\delta_{0} \otimes u)=V_{\lambda}(\delta_{0} \otimes u) +\Phi_{\lambda}(\mr{x}_{0}(\lambda)u). 
\]
We use this characterization of $\mr{x}_{0}(\lambda)$. 
In the statement of Theorem $\ref{RF1}$ we remark that the unit vectors $\mr{v}_{+}(\lambda)$, $\mr{v}_{-}(\lambda)$ 
form a basis in $\mb{C}^{2}$. Indeed if these are related as $\mr{v}_{+}(\lambda)=c \mr{v}_{-}(\lambda)$ with $c \in S^{1}$, 
then $\Phi_{\lambda}(\mr{v}_{+}(\lambda)) \in \ell^{2}(\mb{Z},\mb{C}^{2})$ and 
this is a contradiction since we assume $\lambda \not\in S^{1}$. 
Hence the matrix
\[
\begin{bmatrix}
1 & - \ol{m(\lambda)}  \\
m(\lambda)  & -1
\end{bmatrix}
\]
is non-singular. 
We define $a_{L}, b_{L}, a_{R},b_{R} \in \mb{C}$ by 
\begin{equation}\label{CFF1}
\begin{split}
& 
\begin{bmatrix}
a_{L}(\lambda) \\
b_{L}(\lambda)
\end{bmatrix} 
=\lambda^{-1}
\begin{bmatrix}
1 & - \ol{m(\lambda)}  \\
m(\lambda)  & -1 
\end{bmatrix}^{-1}
\begin{bmatrix}
\ispa{\mr{z}_{L}(0)e_{L}, \mr{v}_{+}(\lambda)} \\
\ispa{\mr{z}_{L}(0)e_{L}, \mr{v}_{-}(\lambda)}
\end{bmatrix},
\\
& \begin{bmatrix}
a_{R}(\lambda) \\
b_{R}(\lambda)
\end{bmatrix} 
=-\lambda^{-1}
\begin{bmatrix}
1 & - \ol{m(\lambda)}  \\
m(\lambda)  & -1
\end{bmatrix}^{-1}
\begin{bmatrix}
\ispa{\mr{z}_{R}(0)e_{R}, \mr{v}_{+}(\lambda)} \\
\ispa{\mr{z}_{R}(0)e_{R}, \mr{v}_{-}(\lambda)}
\end{bmatrix}, 
\end{split}
\end{equation}
where $m(\lambda)=\ispa{\mr{v}_{+}(\lambda),\mr{v}_{-}(\lambda)}_{\mb{C}^{2}}$. 
Then we show 
\begin{equation}\label{VX0}
\begin{split}
\mr{x}_{0}(\lambda)e_{L} & =b_{L}(\lambda)\mr{v}_{-}(\lambda)= 
-\lambda^{-1}  \mr{z}_{L}(0)e_{L} +a_{L}(\lambda)\mr{v}_{+}(\lambda),\\
\mr{x}_{0}(\lambda)e_{R}& = a_{R}(\lambda) \mr{v}_{+}(\lambda) = 
-\lambda^{-1} \mr{z}_{R}(0)e_{R}+b_{R}(\lambda) \mr{v}_{-}(\lambda). 
\end{split}
\end{equation}
The definition $\eqref{CFF1}$ shows 
\[
b_{L}(\lambda) \mr{v}_{-}(\lambda) +\lambda^{-1} \mr{z}_{L}(0)e_{L} =a_{L}(\lambda) \mr{v}_{+}(\lambda), \quad 
a_{R}(\lambda) \mr{v}_{+}(\lambda) +\lambda^{-1} \mr{z}_{R}(0)e_{R}=b_{R}(\lambda) \mr{v}_{-}(\lambda). 
\]
Therefore, the vector $w =b_{L}(\lambda)\mr{v}_{-}(\lambda)$ satisfies 
\[
\sum_{x \geq 1}\|F_{\lambda}(x)[w+\lambda^{-1}\mr{z}_{L}(0)e_{L}]\|_{\mb{C}^{2}}^{2}<+\infty,\quad 
\sum_{x \leq -1}\|F_{\lambda}(x)[w+\lambda^{-1}\mr{z}_{R}(0)e_{L}\|_{\mb{C}^{2}}^{2}<+\infty
\]
since we have $\mr{z}_{R}(0)e_{L}=0$ and the vector $\mr{v}_{-}(\lambda)$ satisfies $\eqref{VPM1}$. 
Then the remark given above gives  $b_{L}(\lambda)\mr{v}_{-}(\lambda)=\mr{x}_{0}(\lambda)e_{L}$, 
which shows the first line of $\eqref{VX0}$. The second line in $\eqref{VX0}$ can be proved similarly. 
Now, we have 
\[
1-|m(\lambda)|^{2}=|\det[\mr{v}_{+}(\lambda) \, \mr{v}_{-}(\lambda)]|^{2}. 
\]
By using this and the definition of $b_{L}$ and $a_{R}$ gives 
\begin{equation}\label{BA}
\lambda b_{L}(\lambda)=-\frac{\ispa{\mr{z}_{L}(0)e_{L}, \mr{v}_{+}(\lambda)^{\perp}}_{\mb{C}^{2}}} 
{\ispa{\mr{v}_{-}(\lambda), \mr{v}_{+}(\lambda)^{\perp}}_{\mb{C}^{2}}},\quad 
\lambda a_{R}(\lambda) =-\frac{\ispa{\mr{z}_{R}(0)e_{R}, \mr{v}_{-}(\lambda)^{\perp}}_{\mb{C}^{2}}} 
{\ispa{\mr{v}_{+}(\lambda), \mr{v}_{-}(\lambda)^{\perp}}_{\mb{C}^{2}}}, 
\end{equation}
which completes the proof. 
\hfill$\square$
\section{Integral representation for the Green kernel}\label{INTEGRAL}
As is seen in the previous sections, the function $\mr{x}_{0} \colon \mb{C} \setminus \sigma(U) \to \mr{M}_{2}(\mb{C})$ 
introduced in $\eqref{X0}$ plays one of central roles in the series of our results. 
But since 
\begin{equation}\label{CONX}
\mr{x}_{0}(\lambda)^{*}=-\lambda^{*} (I+\lambda^{*} \mr{x}_{0}(\lambda^{*})) \quad 
(\lambda \in \mb{C}^{\times} \setminus \sigma(U)), 
\end{equation}
it does not seem that the real-part $[\mr{x}_{0}(\lambda) +\mr{x}_{0}(\lambda)^{*}]/2$ has 
nice properties as the m-Carath\'{e}odory functions have. 
The definition of m-Carath\'{e}odory functions is given in the statement of Lemma $\ref{mC1}$ below. 
Instead of using $\mr{x}_{0}$, we use the following function 
\begin{equation}\label{MX0}
\mr{x}(\lambda)=I+2\lambda \mr{x}_{0}(\lambda).  
\end{equation}
By $\eqref{CONX}$, we have 
\begin{equation}\label{BVXX}
\re \mr{x}(\lambda) :=\frac{1}{2} [\mr{x}(\lambda)+\mr{x}(\lambda)^{*}]=
\lambda \mr{x}_{0}(\lambda) - \lambda^{*} \mr{x}_{0}(\lambda^{*}), 
\end{equation}
which shows the second equality in $\eqref{xmat3}$. 
\begin{lem}\label{mC1}
The function $\mr{x} \colon \mb{C} \setminus \sigma(U) \to \mr{M}_{2}(\mb{C})$ is an m-Caratheodory 
function in the sense that it is a holomorphic function on the unit disc with positive real-part and $\mr{x}(0)=I$. 
\end{lem}
\begin{proof}
It is obvious that $\mr{x}$ is holomorphich on the unit disc and $\mr{x}(0)=I$. We set 
\begin{equation}\label{KO1}
K(\lambda)=(U-\lambda)^{-1}(U+\lambda) \quad (\lambda \in \mb{C} \setminus \sigma(U)). 
\end{equation}
For $\lambda \in \mb{C} \setminus \sigma(U)$ and $u ,v \in \mb{C}^{2}$, we have 
\[
\ispa{\mr{x}_{0}(0)u,v}_{\mb{C}^{2}} =\ispa{R_{\lambda}(0,0)u,v}_{\mb{C}^{2}} 
=\ispa{R(\lambda)(\delta_{0} \otimes u),\delta_{0} \otimes v},  
\]
and hence 
\[
\ispa{\mr{x}(\lambda)u,v}_{\mb{C}^{2}}=\ispa{[I+2\lambda R(\lambda)] \delta_{0} \otimes u, \delta_{0} \otimes v}
=\ispa{K(\lambda) \delta_{0} \otimes u,\delta_{0} \otimes v}. 
\]
This shows that,  
\[
\begin{split}
\ispa{\mr{x}(\lambda)^{*} u,v}_{\mb{C}^{2}} & = \ispa{u,\mr{x}(\lambda)v}_{\mb{C}^{2}} 
=\ol{\ispa{\mr{x}(\lambda)v,u}_{\mb{C}^{2}}} = \ol{\ispa{K(\lambda) \delta_{0} \otimes v, \delta_{0} \otimes u}} 
= \ispa{\delta_{0} \otimes u, K(\lambda) \delta_{0} \otimes v} \\
& = \ispa{K(\lambda)^{*}(\delta_{0} \otimes u),\delta_{0} \otimes v} 
= \ispa{[(U^{*}+\ol{\lambda})(U^{*}-\ol{\lambda})^{-1}] \delta_{0} \otimes u,\delta_{0} \otimes v}
\end{split}
\]
Therefore, the real-part $\re \mr{x}(\lambda)$ of $\mr{x}(\lambda)$ satisfies 
\[
\begin{split}
\re{\ispa{\mr{x}(\lambda)u,u}_{\mb{C}^{2}}} & = \ispa{[\re \mr{x}(\lambda)]u,u}_{\mb{C}^{2}} 
= (1-|\lambda|^{2}) \ispa{(U-\lambda)^{-1}(U^{*}-\ol{\lambda})^{-1} \delta_{0} \otimes u,\delta_{0} \otimes u} \\
& = (1-|\lambda|^{2}) \|R(\lambda)(\delta_{0} \otimes u)\|^{2} >0
\end{split}
\]
if $|\lambda|<1$ and $u \neq 0$, which completes the proof. 
\end{proof}
Therefore, there exists a unique positive $2\times 2$-matrix-valued measure 
\begin{equation}\label{BV}
\Sigma=
\begin{bmatrix}
\mu_{L} & \ol{\alpha} \\
\alpha & \mu_{R}
\end{bmatrix}, 
\end{equation}
on $S^{1}$, where $\mu_{L}$ and $\mu_{R}$ are probability measures and 
$\alpha$ is a complex measure on $S^{1}$, such that we have the 
following Herglotz representation, 
\begin{equation}\label{BV2}
\mr{x}(\lambda)=\int_{S^{1}} \frac{\zeta +\lambda}{\zeta -\lambda}\,d\Sigma(\zeta),\quad 
\Sigma(S^{1}) =I. 
\end{equation}
We refer the readers to \cite{GZ}, \cite{KoM} for the proof of the above fact. 
The positive-matrix-valued measure $\Sigma$ is characterized by 
\begin{equation}\label{BV3}
\begin{split}
\int_{S^{1}} h(\zeta) \ispa{d\Sigma(\zeta)u,v}_{\mb{C}^{2}}  
& =\lim_{r \uparrow 1} \int_{S^{1}} h(\zeta) \ispa{[\re \mr{x}(r\zeta)] u,v}_{\mb{C}^{2}} \,\dbar \zeta 
\end{split}
\end{equation}
for any continuous function $h$ on $S^{1}$, 
where $\dbar \zeta$ denotes the Lebesgue measure on $S^{1}$ with the unit total mass. 
To prove Theorem $\ref{WTK}$, (1), we consider, for a fixed $f \in C_{0}(\mb{Z},\mb{C}^{2})$, 
the following function; 
\[
H_{f}(\lambda) =\ispa{K(\lambda)f,f} \quad (\lambda \in \mb{C}^{2}, \, |\lambda|<1). 
\]
The function $H_{f}$ is holomorphic on the unit disc and its real-part is 
\begin{equation}\label{Hre}
\re H_{f}(\lambda)=(1-|\lambda|^{2}) \|R(\lambda)f\|^{2},
\end{equation}
and hence $H_{f}$ is a scalar Caratheodory function with $H_{f}(0)=\|f\|^{2}$. 
Therefore, there exists a unique positive measure $\nu_{f}$ on $S^{1}$ such that 
\[
H_{f}(\lambda)=\int_{S^{1}} \frac{\zeta +\lambda}{\zeta -\lambda}\,d\nu_{f}(\zeta), 
\quad \nu_{f}(S^{1}) =\|f\|^{2}, 
\]
and the measure $\nu_{f}$ is characterized by 
\begin{equation}\label{CH1}
d\nu_{f}(\zeta)=\wlim_{r \uparrow 1} \re H_{f}(r\zeta)\,\dbar \zeta, 
\end{equation}
where the right-hand side of the above means the weak$^{*}$-limit of $\re H_{f}(r\zeta) \,\dbar \zeta$. 
In what follows we prove the formula $\eqref{GF3}$ with the matrix-valued measure appeared in $\eqref{BV}$ and $\eqref{BV2}$. 
We need the following. 
\begin{lem}\label{RE1}
We set $h_{f}(\lambda):=H_{f}(\lambda) 
-2\lambda \ispa{\mr{x}_{0}(\lambda)\wh{f}^{\ccal}(\lambda), \wh{f}^{\ccal}(\lambda^{*})}$. 
Then $h_{f}(\lambda)$ is holomorphic on $\mb{C}^{\times}$ and 
we have $\re h_{f}(\zeta)=\|\wh{f}^{\ccal}(\zeta)\|_{\mb{C}^{2}}^{2}$ for $\zeta \in S^{1}$. 
\end{lem}
\begin{proof}
By the identity $K(\lambda)=I+2\lambda R(\lambda)$ and Theorem $\ref{Green}$, we have 
\begin{equation}\label{hf1}
\begin{split}
\frac{1}{2}h_{f}(\lambda) & =  
\sum_{x \in \mb{Z}} \sum_{y \leq x-1} \ispa{F_{\lambda}(x) \mr{z}_{L}(0) F_{\lambda^{*}}(y) f(y),f(x)}_{\mb{C}^{2}} \\
& \hspace{20pt} + \sum_{x \in \mb{Z}} \sum_{y \geq x+1}  
\ispa{F_{\lambda}(x) \mr{z}_{R}(0) F_{\lambda^{*}}(y)^{*} f(y), f(x)}_{\mb{C}^{2}} \\
& \hspace{20pt} + \sum_{x \in \mb{Z}} \ispa{F_{\lambda}(x) \mr{z}_{L}(0) F_{\lambda^{*}}(x)^{*}f(x),f(x)}_{\mb{C}^{2}} 
-\sum_{x \in \mb{Z}} \ispa{\mr{z}_{L}(x) f(x), f(x)}_{\mb{C}^{2}} +\frac{1}{2} \|f\|^{2}. 
\end{split}
\end{equation}
The sums are all finite sums because $f$ has a finite support. In the formula above, 
we note that $\lambda^{*}=\ol{\lambda}^{-1}$ itself is not holomorphic in $\lambda$, but 
the conjugate $F_{\lambda^{*}}(y)^{*}$ is holomorphic in $\lambda \in \mb{C}^{\times}$, 
and hence $h_{f}$ is also holomorphic on $\mb{C}^{\times}$. 
By taking the complex conjugate of the above formula and using the formula (6) in Lemma $\ref{FMS}$ 
and Lemma $\ref{REC1}$, we have 
\begin{equation}\label{hf2}
\begin{split}
\frac{1}{2}\ol{h_{f}(\lambda)} & = \ispa{\wh{f}^{\ccal}(\lambda^{*}), \wh{f}^{\ccal}(\lambda)}_{\mb{C}^{2}} -\frac{1}{2}\|f\|^{2}  \\
& \hspace{20pt} -  \sum_{x \in \mb{Z}} \sum_{y \geq x+1} \ispa{F_{\lambda^{*}}(x) \mr{z}_{R}(0) F_{\lambda}(y)^{*}f(y),f(x)}_{\mb{C}^{2}} \\
& \hspace{20pt} - \sum_{x \in \mb{Z}} \sum_{y \leq x-1} \ispa{F_{\lambda^{*}}(x) \mr{z}_{L}(0) F_{\lambda}(y)^{*} f(y), f(x)}_{\mb{C}^{2}} \\
& \hspace{20pt} -\sum_{x \in \mb{Z}} \ispa{F_{\lambda^{*}}(x) \mr{z}_{L}(0) F_{\lambda}(x)^{*} f(x), f(x)}_{\mb{C}^{2}} 
+\sum_{x \in \mb{Z}} \ispa{\mr{z}_{L}(x) f(x),f(x)}_{\mb{C}^{2}}. 
\end{split}
\end{equation}
We note that in the above, the sign of $\|f\|^{2}$ has changed form that in the formula for $h_{f}(\lambda)/2$. 
This is caused by the use of the formula $\mr{z}_{L}(x)^{*}=I-\mr{z}_{R}(x)$ in the last line of $\eqref{hf1}$. 
Now, we take $\lambda=\zeta \in S^{1}$. Then since $\zeta^{*}=\zeta$, the first two terms of $\eqref{hf1}$ and $\eqref{hf2}$ 
canceled after taking the sum of these two formulas. Therefore, we obtain the assertion. 
\end{proof}
\par
%\vspace{10pt}
\medskip
\noindent{\it Proof of Theorem $\ref{WTK}$.}\hspace{3pt} 
We prove in the following that 
\begin{equation}\label{GFP}
d\nu_{f}(\zeta)=\ispa{d\Sigma(\zeta)\wh{f}^{\ccal}(\zeta), \wh{f}^{\ccal}(\zeta)}_{\mb{C}^{2}},  
\end{equation}
which means 
\begin{equation}\label{GFP2}
\ispa{K(\lambda)f,f}=\int_{S^{1}} \frac{\zeta +\lambda}{\zeta -\lambda} 
\ispa{d\Sigma(\zeta)\wh{f}^{\ccal}(\zeta), \wh{f}^{\ccal}(\zeta)}_{\mb{C}^{2}}. 
\end{equation}
Setting $\lambda=0$ in $\eqref{GFP2}$ shows 
\begin{equation}\label{GFP3}
\ispa{f,f}=\int_{S^{1}}\ispa{d\Sigma(\zeta)\wh{f}^{\ccal}(\zeta), \wh{f}^{\ccal}(\zeta)}_{\mb{C}^{2}}, 
\end{equation}
and hence the formula $\eqref{GF3}$ in Theorem $\ref{WTK}$ with $f=g$ follows. 
Then the polarization identity shows $\eqref{GF3}$. 
We note that the uniqueness in the statement (1) in Theorem $\ref{WTK}$ follows from 
the uniqueness of the positive-matrix-valued measure in $\eqref{BV2}$. Thus, it is enough to prove $\eqref{GFP}$.  
The function $H_{f}(\lambda)$ can be written as 
\[
H_{f}(\lambda) =\ispa{\mr{x}(\lambda)\wh{f}^{\ccal}(\lambda), \wh{f}^{\ccal}(\lambda^{*})}_{\mb{C}^{2}} 
-\ispa{\wh{f}^{\ccal}(\lambda), \wh{f}^{\ccal}(\lambda^{*}) }_{\mb{C}^{2}} +h_{f}(\lambda), 
\]
where $h_{f}$ is defined in Lemma $\ref{RE1}$. Therefore, by $\eqref{CH1}$, we have 
\begin{equation}\label{CH2}
d\nu_{f}(\zeta) = \wlim_{r \uparrow 1} \re \ispa{\mr{x}(r\zeta) \wh{f}^{\ccal} (r\zeta), \wh{f}^{\ccal}(r^{-1}\zeta)}_{\mb{C}^{2}} \,\dbar \zeta. 
\end{equation}
The matrix-valued function $\mr{x}(\lambda)$ could have singularity on $S^{1}$ 
and thus we perform computation on the right-hand side above a bit carefully as follows. 
For simplicity of notation we set 
\[
F_{L}(\lambda)=\ispa{\wh{f}^{\ccal}(\lambda),e_{L}}_{\mb{C}^{2}},\quad 
F_{R}(\lambda)=\ispa{\wh{f}^{\ccal}(\lambda),e_{R} }_{\mb{C}^{2}}. 
\]
Then $\re \ispa{\mr{x}(\lambda)\wh{f}^{\ccal}(\lambda), \wh{f}^{\ccal}(\lambda^{*}) }$ is written as 
\begin{equation}\label{CH4}
\begin{split}
2 \re \ispa{\mr{x}(\lambda)\wh{f}^{\ccal}(\lambda), \wh{f}^{\ccal}(\lambda^{*}) } 
& =  F_{L}(\lambda) \ol{F_{L}(\lambda^{*})} \ispa{\mr{x}(\lambda) e_{L},e_{L}}_{\mb{C}^{2}} 
    +F_{L}(\lambda^{*}) \ol{F_{L}(\lambda)} \ispa{\mr{x}(\lambda)^{*}e_{L},e_{L}}_{\mb{C}^{2}}  \\
   & \hspace{10pt} +F_{R}(\lambda) \ol{F_{L}(\lambda^{*})} \ispa{\mr{x}(\lambda)e_{R},e_{L}}_{\mb{C}^{2}} 
+F_{R}(\lambda^{*}) \ol{F_{L}(\lambda)} \ispa{\mr{x}(\lambda)^{*}e_{R},e_{L}}_{\mb{C}^{2}} \\
& \hspace{10pt} +F_{L}(\lambda) \ol{F_{R} (\lambda^{*})} \ispa{\mr{x}(\lambda)e_{L},e_{R}}_{\mb{C}^{2}}   
   + F_{L}(\lambda^{*})\ol{F_{R}(\lambda)} \ispa{\mr{x}(\lambda)^{*}e_{L},e_{R}}_{\mb{C}^{2}}  \\
& \hspace{10pt} + F_{R}(\lambda) \ol{F_{R}(\lambda^{*})} \ispa{\mr{x}(\lambda) e_{R}, e_{R}}_{\mb{C}^{2}} 
   + F_{R}(\lambda^{*}) \ol{F_{R}(\lambda)} \ispa{\mr{x}(\lambda)^{*}e_{R},e_{R}}_{\mb{C}^{2}}.     
\end{split}
\end{equation}
For example, we write $g(\lambda)$ for the first line of the above, that is, 
\[
g(\lambda)=F_{L}(\lambda) \ol{F_{L}(\lambda^{*})} \ispa{\mr{x}(\lambda) e_{L},e_{L}}_{\mb{C}^{2}} 
   +F_{L}(\lambda^{*}) \ol{F_{L} (\lambda)} \ispa{\mr{x}(\lambda)^{*}e_{L},e_{L}}_{\mb{C}^{2}}. 
\] 
The function $g(\lambda)$ is smooth on $\mb{C}^{\times} \setminus \sigma(U)$. 
For $0<r<1$ and $\zeta \in S^{1}$, we set 
\[
k(r,\zeta)=F_{L}(r^{-1}\zeta)-F_{L}(r\zeta)
\]
so that 
\begin{equation}\label{gA1} 
\begin{split}
g(r\zeta) & =F_{L}(r\zeta) \ol{k(r,\zeta)} \ispa{\mr{x}(r\zeta) e_{L},e_{L}}_{\mb{C}^{2}} +
\ol{F_{L}(r\zeta)} k(r,\zeta) \ispa{\mr{x}(r\zeta)^{*}e_{L},e_{L} }_{\mb{C}^{2}}  \\
& \hspace{20pt} + 2|F_{L}(r\zeta)|^{2} \ispa{ [\re \mr{x}(r\zeta)] e_{L},e_{L}}_{\mb{C}^{2}}.
\end{split}
\end{equation}
The last term in the above converges weakly to $2|F_{L}(\zeta)|^{2} \ispa{d\Sigma(\zeta)e_{L},e_{L}}$. 
Since $F_{L}(\lambda)$ is holomorphic on $\mb{C}^{\times}$, 
we have 
\[
|k(r,\zeta)| =|F_{L}(r^{-1}\zeta) -F_{L} (r\zeta)| \leq 
2\|F_{L}' \|_{L^{\infty}(A)} |1-r^{2}|, 
\]
where $A=\{\lambda \in \mb{C} \mid 1/2 \leq |\lambda| \leq 3/2\}$. 
Then for any continuous function $h(\zeta)$ on $S^{1}$ and $1/2 \leq r<1$, we have 
\[
\begin{split} 
& 
\left|
\int_{S^{1}} h(\zeta) F_{L}(r\zeta) \ol{k(r,\zeta)} \ispa{\mr{x}(r\zeta) e_{L},e_{L}}_{\mb{C}^{2}} \dbar \zeta 
\right| \\
& \leq C \|F_{L}' \|_{L^{\infty}(A)}  \|F_{L}\|_{L^{\infty}(S^{1})}   \|h\|_{L^{\infty}(S^{1})}   |1-r| 
\int_{S^{1}} \|\mr{x}(r\zeta) e_{L}\|_{\mb{C}^{2}} \dbar \zeta
\end{split}
\]
with a constant $C>0$. 
For any $u \in \mb{C}^{2}$, $\mr{x}_{0}(\lambda)u=R(\lambda)(\delta_{0} \otimes u)(0)$ and hence 
\[
\|\mr{x}_{0}(\lambda)u\|_{\mb{C}^{2}}^{2} = 
\|R(\lambda) (\delta_{0} \otimes u)(0) \|_{\mb{C}^{2}}^{2} \leq \|R(\lambda) (\delta_{0} \otimes u)\|^{2}
=(1-|\lambda|^{2})^{-1} \re H_{\delta_{0} \otimes u}(\lambda). 
\]
We note that $\mr{x}(\lambda) = I+2I \mr{x}_{0}(\lambda)$, and hence, for $|\lambda|<1$, 
\[
\|\mr{x}(\lambda)e_{L}\|_{\mb{C}^{2}} \leq 1+2\|\mr{x}_{0}(\lambda)e_{L}\|_{\mb{C}^{2}}
\leq 1+2(1-|\lambda|^{2})^{-1/2} [\re H_{\delta_{0} \otimes u}(\lambda)]^{1/2}. 
\]
By $\eqref{CH1}$, the integral 
\[
(1-r)(1-r^{2})^{-1/2} \int_{S^{1}} [\re H_{\delta_{0} \otimes u}(r\zeta)]^{1/2} \,\dbar \zeta 
\leq C(1-r)^{1/2} \left( \int_{S^{1}} \re H_{\delta_{0} \otimes u}(r\zeta) \,\dbar \zeta \right)^{1/2}
\]
tends to zero as $r \uparrow 1$. Therefore we obtain 
\[
\begin{split}
\lim_{r \uparrow 1}\int_{S^{1}} h(\zeta) g(r\zeta) \,\dbar \zeta & =
2\lim_{r \uparrow 1} \int_{S^{1}} h(\zeta) |F_{L}(r\zeta)|^{2} \ispa{[\re \mr{x}(r\zeta)] e_{L},e_{L}}_{\mb{C}^{2}} \,\dbar \zeta  \\
& =2\int_{S^{1}} h(\zeta) |F_{L}(\zeta)|^{2} 
\ispa{d\Sigma(\zeta)e_{L},e_{L}}_{\mb{C}^{2}}
\end{split}
\]
Similar computation works well for each line in $\eqref{CH4}$ and hence by $\eqref{CH2}$ we conclude
\begin{equation}\label{CH5}
\begin{split}
\int_{S^{1}} h(\zeta) d\nu_{f}(\zeta) & = 
\int_{S^{1}} h(\zeta) \left(|F_{L}(\zeta)|^{2} \ispa{d\Sigma(\zeta)e_{L},e_{L}}_{\mb{C}^{2}} 
+F_{R}(\zeta)\ol{F_{L}(\zeta)} \ispa{d\Sigma(\zeta)e_{R},e_{L}}_{\mb{C}^{2}} \right.\\
& \hspace{20pt} \left. +F_{L}(\zeta)\ol{F_{R}(\zeta)} \ispa{d\Sigma(\zeta)e_{L},e_{R}}_{\mb{C}^{2}} 
+|F_{R}(\zeta)|^{2} \ispa{d\Sigma(\zeta)e_{R},e_{R}}_{\mb{C}^{2}} \right) \\
& = \int_{S^{1}} h(\zeta) \ispa{d\Sigma(\zeta) \wh{f}^{\ccal} (\zeta), \wh{f}^{\ccal} (\zeta) }_{\mb{C}^{2}}, 
\end{split}
\end{equation}
which completes the proof of $\eqref{GFP}$ and hence Theorem $\ref{WTK}$, (1). 
The assertion (2) in Theorem $\ref{WTK}$ follows from $\eqref{GFP3}$. 
We remark that the equation $\eqref{CH5}$ and the formula $\eqref{QWF2}$ show 
\[
\ispa{h(U)f,g}=\int_{S^{1}} h(\zeta) \ispa{d\Sigma(\zeta) \wh{f}^{\ccal}(\zeta), \wh{g}^{\ccal}(\zeta)}_{\mb{C}^{2}}
\]
for any continuous function $h$ on $S^{1}$. Therefore, the uniqueness of the spectral measure shows 
\[
\ispa{E(A)f,g}=\int_{A} \ispa{d\Sigma(\zeta) \wh{f}^{\ccal}(\zeta), \wh{g}^{\ccal}(\zeta)}_{\mb{C}^{2}}
\]
for any Borel set $A$ in $S^{1}$. 
Taking $g=\delta_{x} \otimes u$ with $x \in \mb{Z}$ and $u \in \mb{C}^{2}$, and 
noting $\fcal_{\ccal}[\delta_{x} \otimes u](\lambda)=F_{\lambda^{*}}(x)^{*}u$, we have 
\[
\ispa{[E(A)f](x),u}_{\mb{C}^{2}}= \int_{A} \ispa{d\Sigma(\zeta) \wh{f}^{\ccal}(\zeta), F_{\zeta}(x)^{*}u}_{\mb{C}^{2}}
= \int_{A} \ispa{F_{\zeta}(x)d\Sigma(\zeta) \wh{f}^{\ccal}(\zeta), u}_{\mb{C}^{2}}. 
\]
From this we conclude $\eqref{Sk1}$. 
\hfill$\square$
\par
%\vspace{10pt}
\medskip
\noindent{\it Proof of Theorem $\ref{FTT}$.} From $\eqref{PF1}$, it is enough to show that 
the subspace $\fcal_{\ccal}[C_{0}(\mb{Z},\mb{C}^{2})]$ is dense in $L^{2}(S^{1},\mb{C}^{2})_{\Sigma}$. 
Suppose that $k \in L^{2}(S^{1},\mb{C}^{2})_{\Sigma}$ is perpendicular 
to $\wh{f}^{\ccal}$ for any $f \in C_{0}(\mb{Z},\mb{C}^{2})$. For any $u \in \mb{C}^{2}$ and $n \in \mb{Z}$, 
we have, by $\eqref{QWF2}$,  $\wh{U^{n} (\delta_{0} \otimes u)}^{\ccal}(\lambda)=\lambda^{n}u$. 
Therefore, $k$ is perpendicular to any $\mb{C}^{2}$-valued Laurent polynomial. 
The set of all $\mb{C}^{2}$-valued Laurent polynomials is dense in the space $C(S^{1},\mb{C}^{2})$ 
of $\mb{C}^{2}$-valued continuous functions on $S^{1}$ with the norm 
$\dsp \|p\|_{\infty}=\sup_{\zeta \in S^{1}}\|p(\zeta)\|_{\mb{C}^{2}}$. Furthermore, 
the norm $\dsp \|p\|_{\Sigma}^{2}=\int_{S^{1}} \ispa{d\Sigma(\zeta)p(\zeta),p(\zeta)}_{\mb{C}^{2}}$ satisfies 
the inequality 
\[
\| p\|_{\Sigma}^{2} \leq 4\|p\|_{\infty}^{2}. 
\]
Let $f \in C(S^{1},\mb{C}^{2})$ and let $\{p_{n}\}$ be a sequence of $\mb{C}^{2}$-valued 
Laurent polynomials such that $\|f-p_{n}\|_{\infty} \to 0$ as $n \to \infty$. Then 
\[
|\ispa{k,f}_{\Sigma}| \leq |\ispa{k,p_{n}}_{\Sigma}| +|\ispa{k,f-p_{n}}_{\Sigma}| \leq \|k\|_{\Sigma} \|f-p_{n}\|_{\Sigma}
\leq  4\|k\|_{\Sigma} \|f-p_{n}\|_{\infty} \to 0 \ \ (n \to \infty), 
\]
which shows that $k$ is perpendicular to $C(S^{1},\mb{C}^{2})$ in $L^{2}(S^{1},\mb{C}^{2})_{\Sigma}$. 
Since $C(S^{1},\mb{C}^{2})$ is dense in $L^{2}(S^{1},\mb{C}^{2})_{\Sigma}$ we conclude $k=0$. 
Now since $C_{0}(\mb{Z},\mb{C}^{2})$ is dense in $\ell^{2}(\mb{Z},\mb{C}^{2})$, 
the equation $\eqref{QWF2}$ still holds for any $f \in \ell^{2}(\mb{Z},\mb{C}^{2})$. 
This completes the proof of Theorem $\ref{FTT}$. 
\hfill$\square$
\par
\medskip
\section{Examples}\label{EXS}
In the final section, we consider two examples, homogeneous quantum walks and 
certain quantum walks with non-constant coin matrix. The quantum walks with non-constant 
coin matrix considered here is a special case of the so-called two phase model discussed originally in \cite{EKST}. 
In these examples we compute concretely the positive-matrix-valued measure $\Sigma$ 
as a boundary value of the matrix-valued function $\re \mr{x}(\lambda)$. 
\par
\medskip
\noindent{\rm (1)\ {\it Homogeneous quantum walks}.}\hspace{3pt}
First of all, let us consider the fundamental example, namely the case where 
the coin matrix $\ccal$ is constant, say $\ccal(x)=C \in \mr{U}(2)$ for any $x \in \mb{Z}$. 
For simplicity we assume 
\[
C \in \mr{SU}(2)
\]
so that $C$ can be written as 
\[
C=
\begin{bmatrix}
\alpha & \beta \\
-\ol{\beta} & \ol{\alpha}
\end{bmatrix}=
\quad |\alpha|^{2}+|\beta|^{2}=1.  
\]
The transfer matrix $T_{\lambda}(x)$ does not depend on $x \in \mb{Z}$ and is denoted by $T_{C}(\lambda)$, 
which and whose inverse are given by 
\[
T_{C}(\lambda) =
\begin{bmatrix}
{\dsp \frac{1}{\alpha}(\lambda+\lambda^{-1}|\beta|^{2}) }&{\dsp  -\lambda^{-1}\frac{\beta \ol{\alpha}}{\alpha} }\\
{\dsp -\lambda^{-1}\ol{\beta} }&{\dsp \lambda^{-1}\ol{\alpha} }
\end{bmatrix}, 
\quad 
T_{C}(\lambda)^{-1} 
=
\begin{bmatrix}
{\dsp \lambda^{-1} \alpha }&{\dsp \lambda^{-1}\beta }\\
{\dsp \lambda^{-1} \frac{\alpha \ol{\beta}}{\ol{\alpha}} }
&{\dsp \frac{1}{\ol{\alpha}} (\lambda+\lambda^{-1}|\beta|^{2})}
\end{bmatrix}.
\]
The matrix $F_{\lambda}(x)$ is then given by $F_{\lambda}(x)=T_{C}(\lambda)^{x}$ for any $x \in \mb{Z}$. 

When $\beta=0$, we have $|\alpha|=1$ and the transfer matrix $T_{C}(\lambda)$ is a diagonal matrix,  
and $F_{\lambda}(x)$ is given by 
\begin{equation}
F_{\lambda}(x)=T_{C}(\lambda)^{x}=
\begin{bmatrix}
(\lambda/\alpha)^{x} & 0 \\[7pt]
0 & 1/(\alpha\lambda)^{x}
\end{bmatrix}. 
\end{equation}
In this case we can take
\[
\mr{v}_{+}(\lambda)=e_{L},\quad \mr{v}_{-}(\lambda)=e_{R}, 
\]
for the unit vectors in Theorem $\ref{RF1}$, which are the eigenvectors of $T_{C}(\lambda)$ with eigenvalues 
$z_{+}(\lambda)=\lambda/\alpha$, $z_{-}(\lambda)=1/(\alpha \lambda)$, respectively. 
We have $\mr{z}_{L}(0)e_{L}=e_{L}$, $\mr{z}_{R}(0)e_{R}=e_{R}$, 
and hence $\mr{x}_{0}(\lambda)=0$, $\mr{x}(\lambda)=I$ by $\eqref{xmat}$. 
Therefore, by using the characterization $\eqref{BV3}$ of $\Sigma$, 
it is concluded that the positive matrix-valued measure $\Sigma$ is given by $d\Sigma(\zeta)=I \dbar \zeta$, 
identity matrix times normalized Lebesgue measure $\dbar \zeta$. In this case the corresponding 
QW-Fourier transform basically coincides with the usual Fourier series expansion. 

Next we consider the case $\beta \neq 0$. Hence $|\alpha|<1$. 
The characteristic equation for $T_{C}(\lambda)$ is 
\begin{equation}\label{CHE1}
z^{2}-\frac{2}{\alpha} J(\lambda) z +\frac{\ol{\alpha}}{\alpha} =0,\quad J(\lambda)=\frac{\lambda+\lambda^{-1}}{2}. 
\end{equation}
It is easy to show that if $|\lambda|<1$ then $T_{C}(\lambda)$ does not have eigenvalues in $S^{1}$ and 
$T_{C}(\lambda)$ has an eigenvalue with multiplicity two if and only if $\lambda$ is one of the following four points
\[
|\alpha| +i |\beta|,\quad -|\alpha|+i|\beta|,\quad 
-|\alpha|-i|\beta|,\quad |\alpha|-i|\beta|. 
\]
Thus, for $\lambda$ satisfying $0<\lambda<1$, the eigenvalues of $T_{C}(\lambda)$ can be labeled as $z_{\pm}(\lambda)$ 
satisfying $|z_{+}(\lambda)|<1<|z_{-}(\lambda)|$ and $z_{\pm}(\lambda)$ are holomorphic in $\lambda$. 
Concretely $z_{+}(\lambda)$ is given by 
\begin{equation}\label{SOL11}
\alpha z_{+}(\lambda)=
\begin{cases}
{\dsp J(\lambda)-\sqrt{J(\lambda)^{2}-|\alpha|^{2}} }& (|\lambda|<1,\, \re \lambda>0) \ \text{or} \ (|\lambda|>1,\, \re \lambda>0) \\[7pt]
{\dsp J(\lambda)+\sqrt{J(\lambda)^{2}-|\alpha|^{2}} }& (|\lambda|<1,\, \re \lambda <0) \ \text{or} \ (|\lambda|>1,\, \re \lambda<0) \\[7pt] 
{\dsp J(\lambda)+i \sqrt{|\alpha|^{2}-J(\lambda)^{2}} }& (|\lambda|<1,\, \im \lambda>0) \ \text{or} \ (|\lambda|>1,\, \im \lambda<0) \\[7pt]
{\dsp J(\lambda)-i\sqrt{|\alpha|^{2}-J(\lambda)^{2}} }& (|\lambda|<1,\, \im \lambda<0) \ \text{or} \ (|\lambda|>1,\, \im \lambda>0)
\end{cases}
\end{equation}
For the case of the constant coin, it is rather easy to use usual Fourier series 
\[
\fcal \colon \ell^{2}(\mb{Z},\mb{C}^{2}) \to L^{2}(S^{1},\mb{C}^{2}),\quad 
\fcal[f](z)=\sum_{x \in \mb{Z}}f(x) z^{x}, 
\] 
to compute the matrix $\mr{x}_{0}(\lambda)$. Indeed, we have
\begin{equation}\label{RESC1}
R_{\lambda}(x,y)=\int_{S^{1}} z^{y-x} \left(
\wh{U}(z) -\lambda
\right)^{-1}\,\dbar z, 
\end{equation}
where $\dbar z=dz/(2\pi i z)$ as before and $\wh{U}(z)$ is the matrix-valued function given by 
\[
\wh{U}(z)=
\begin{bmatrix}
\alpha z^{-1} & \beta z^{-1} \\
-\ol{\beta} z & \ol{\alpha} z
\end{bmatrix}. 
\]
Performing the integral in $\eqref{RESC1}$ for $x=y=0$, we see 
\begin{equation}\label{COL}
\mr{x}_{0}(\lambda) =
\frac{1}{\lambda \alpha (z_{+}(\lambda) -z_{-}(\lambda))} 
\begin{bmatrix}
{\dsp \lambda -\alpha z_{+}(\lambda) }
&{\dsp \beta z_{+}(\lambda) }\\[10pt]
-{\dsp \frac{\ol{\beta}}{\ol{\alpha}} \alpha z_{+}(\lambda) } & 
{\dsp \lambda-\alpha z_{+}(\lambda) }
\end{bmatrix}. 
\end{equation}
The formula $\eqref{COL}$ can also be deduced by using Corollary $\ref{RF1}$, 
although it needs somehow complicated computation. 
Indeed, we can take the vectors $\mr{v}_{+}(\lambda)$, $\mr{v}_{-}(\lambda)$ as 
eigenvectors of $z_{\pm}(\lambda)$. Explicitly, we set 
\begin{equation}\label{EXV1}
\mr{v}_{+}(\lambda)=\frac{1}{\sqrt{|\lambda z_{+}(\lambda) -\ol{\alpha}|^{2}+|\beta|^{2}}} 
\begin{bmatrix}
\lambda z_{+}(\lambda) -\ol{\alpha} \\
-\ol{\beta}
\end{bmatrix},\quad 
\mr{v}_{-}(\lambda)=\frac{1}{\sqrt{|\lambda z_{-}(\lambda) -\ol{\alpha}|^{2}+|\beta|^{2}}} 
\begin{bmatrix}
\lambda z_{-}(\lambda) -\ol{\alpha} \\
-\ol{\beta}
\end{bmatrix}. 
\end{equation}
A direct computation is used to check that, when $\lambda$ is real and $0<\lambda<1$, 
the vectors $\mr{v}_{+}(\lambda)$ and $\mr{v}_{-}(\lambda)$ form an 
orthonormal basis in $\mb{C}^{2}$ and we have, by Theorem $\ref{RF1}$, 
or rather by the formula $\eqref{VX0}$ in the proof of Theorem $\ref{RF1}$, 
\begin{equation}\label{MT1}
\begin{split} 
\mr{x}_{0}(\lambda)e_{L} & =
-\frac{\lambda^{-1} (\ol{\lambda} \ol{z_{-}(\lambda)} -1/\ol{\alpha})}
{|\lambda z_{-}(\lambda) -\ol{\alpha}|^{2}+|\beta|^{2}}  
\begin{bmatrix}
\lambda z_{-}(\lambda) -\ol{\alpha} \\
-\ol{\beta}
\end{bmatrix} \\
\mr{x}_{0}(\lambda) e_{R} & = 
\frac{\alpha^{-1}\beta \lambda^{-1} \ol{\lambda} \ol{\mr{z_{+}(\lambda)}}}
{|\lambda z_{+}(\lambda) -\ol{\alpha}|^{2}+|\beta|^{2}} 
\begin{bmatrix}
\lambda z_{+}(\lambda) -\ol{\alpha} \\
-\ol{\beta}
\end{bmatrix}
\end{split}
\end{equation}
That the two formulas $\eqref{MT1}$ and $\eqref{COL}$ are identical for $0<\lambda<1$ 
can be verified by the formula 
\[
\begin{split}
 |\lambda z_{+}(\lambda) -\ol{\alpha}|^{2}+|\beta|^{2}
& = \frac{\lambda \alpha}{\ol{\alpha}} (z_{+}-z_{-}) (\lambda z_{+} -\ol{\alpha} ),  \\
|\lambda z_{-}(\lambda) -\ol{\alpha}|^{2}+|\beta|^{2} 
& = \frac{\lambda \alpha}{\ol{\alpha}} (z_{+}-z_{-}) (\ol{\alpha} - \lambda z_{-}  ),  
\end{split}
\]
which is valid for $0<\lambda<1$, and then one can use the analytic continuation for $0<\re \lambda <1$. 
We note that from the formula $\eqref{RESC1}$ the Green kernel must satisfy 
\[
R_{\lambda}(x,x)=R_{\lambda}(0,0)=\mr{x}_{0}(\lambda), 
\]
which is, according to $\eqref{GF2}$ in Theorem $\ref{Green}$, equivalent to 
\begin{equation}\label{GFG1}
T_{C}(\lambda) [\mr{x}_{0}(\lambda)+\mr{z}_{L}(0)] T_{C}(\lambda^{*})^{*}=[\mr{x}_{0}(\lambda)+\mr{z}_{L}(0)]
\end{equation}
because $\mr{z}_{L}(x)=\mr{z}_{L}(0)$. The equation $\eqref{GFG1}$ can be verified 
directly by $\eqref{COL}$ and the definition of $T_{C}(\lambda)$. 
A straightforward computation shows
\begin{equation}\label{MC1}
\mr{x}(\lambda) =I +2 \lambda \mr{x}_{0}(\lambda) 
=\frac{2}{\alpha (z_{+}(\mu) -z_{-}(\lambda))}
\begin{bmatrix}
{\dsp K(\lambda) }&{\dsp \frac{\beta}{\alpha} \alpha z_{+}(\lambda) } \\
{\dsp -\frac{ \ol{\beta}}{\ol{\alpha}} \alpha z_{+}(\lambda) }&{\dsp  K(\lambda) } 
\end{bmatrix}, \quad 
K(\lambda)=\frac{\lambda -\lambda^{-1}}{2}. 
\end{equation}
From the formula $\eqref{MC1}$ it is easy to show the following. 
\begin{thm}\label{CONT1}
The positive matrix-valued measure $d\Sigma(\zeta)$ is given by 
\[
d\Sigma(\zeta)=\frac{1}{\sqrt{|\alpha|^{2}-\re \zeta ^{2}}}
\begin{bmatrix}
{\dsp |\im (\zeta)| }&{\dsp - i \frac{\beta}{\alpha} \re \zeta }\\
{\dsp i \frac{\ol{\beta}}{\ol{\alpha}}  \re\zeta }&{\dsp  |\im (\zeta)| }
\end{bmatrix}
\chi_{\sigma(U)} \,\dbar \zeta, 
\]
where $\chi_{\sigma(U)}$ is the characteristic function of 
the spectrum $\sigma(U)=\{\zeta \in S^{1} \mid \re\zeta \leq |\alpha|\}$ of $U$. 
\end{thm}
\medskip
\noindent{\rm (2)}\hspace{3pt} {\it Simplest Two-phase model.}\hspace{3pt}
Let $C_{0}$, $C_{\pm}$ be three $2 \times 2$ special unitary matrices, and we write them as 
\[
C_{0} =
\begin{bmatrix}
\alpha_{0} & \beta_{0} \\
-\ol{\beta}_{0} & \alpha_{0}
\end{bmatrix}, \quad 
C_{\pm} =
\begin{bmatrix}
\alpha_{\pm} & \beta_{\pm} \\
-\ol{\beta_{\pm}} & \alpha_{\pm}
\end{bmatrix},\quad 
|\alpha_{0}|^{2}+|\beta_{0}|^{2}=|\alpha_{\pm}|^{2}+|\beta_{\pm}|^{2}=1. 
\]
Let $\ccal \colon \mb{Z} \to \mr{U}(2)$ be a coin matrix defined as 
\[
\ccal(x)=C_{+} \quad (x \geq 1),\quad \ccal(0)=C_{0},\quad \ccal(x)=C_{-}\quad (x \leq -1). 
\]
The quantum walk $U(\ccal)$ defined by the coin matrix $\ccal$ is called a two-phase model with one defect. 
Since the computation is a bit complicated, 
we assume the following strong assumptions. 
\begin{equation}\label{AssA1}
\begin{gathered}
0<\rho:=|\alpha_{+}|=|\alpha_{-}|<1,\quad \beta:=\beta_{+}=\beta_{-}, \quad 
|\re (\beta_{0}\beta)| < |\beta|^{2}. 
\end{gathered}
\end{equation}
Therefore, $C_{+}$ and $C_{-}$ differs only by the phase of the diagonals. 
We note that the first two assumptions are for simplifying the computation, 
but the last assumption is imposed for $U$ to have eigenvalues. 
Eigenvalues of much general two-phase models are studied in \cite{KS}. 
In this case the transfer matrix $T_{\lambda}(x)$ satisfies 
\[
T_{\lambda}(x)=
\begin{cases}
T_{C_{+}}(\lambda) &  (x \geq 1), \\
T_{C_{-}}(\lambda) & (x \leq -2), 
\end{cases}
\]
and the matrix-valued function $F_{\lambda}(x)$ satisfies 
\begin{equation}\label{FM11}
F_{\lambda}(x)=
\begin{cases}
T_{C_{+}}(\lambda)^{x-1} T_{\lambda}(0) & (x \geq 1) \\
T_{C_{-}}(\lambda)^{x+1} T_{\lambda}(-1)^{-1} & (x \leq -1)
\end{cases}
\end{equation}
We denote $z_{\pm}(C_{\pm}, \lambda)$ the eigenvalues of $T_{C_{\pm}}(\lambda)$ satisfying
\begin{equation}\label{ORD1}
|z_{+}(C_{\pm}, \lambda)|<1<|z_{-}(C_{\pm}, \lambda)|,\quad \lambda \in \mb{C}^{\times} \setminus S^{1}, 
\end{equation}
and we define vectors $w_{\pm}(C_{\pm}, \lambda)$ by
\begin{equation}\label{2EG}
w_{\pm}(C_{+}, \lambda) =
\begin{bmatrix}
{\dsp \ol{\alpha_{0}} \left( \lambda  z_{\pm}(C_{+}, \lambda) -\ol{\alpha_{+}} \right)}\\
- \ol{\alpha_{0}} \ol{\beta}
\end{bmatrix},\quad 
w_{\pm}(C_{-}, \lambda) =
\begin{bmatrix}
{\dsp \alpha_{0} \alpha_{-}  \left( z_{\pm}(C_{-}, \lambda) -\ol{\alpha_{-}} \lambda^{-1} \right) }\\
- \alpha_{0} \alpha_{-} \ol{\beta} \lambda^{-1}
\end{bmatrix}. 
\end{equation}
The vectors $w_{\pm}(C_{\pm}, \lambda)$ are eigenvectors of $T_{C_{\pm}}(\lambda)$ 
with the eigenvalues $z_{\pm}(C_{\pm}, \lambda)$, respectively. 
By the equations $\eqref{FM11}$, $\eqref{ORD1}$ and the analysis for the case of 
homogeneous quantum walks with a constant coin matrix, 
the unit vectors $\mr{v}_{+}(\lambda)$, $\mr{v}_{-}(\lambda)$ in Theorem $\ref{RF1}$ can be chosen as 
\begin{equation}\label{VPM22}
\begin{split}
\mr{v}_{+}(\lambda) & =\frac{1}{\|T_{\lambda}(0)^{-1}w_{+}(C_{+}, \lambda)\|_{\mb{C}^{2}}} 
T_{\lambda}(0)^{-1}w_{+}(C_{+}, \lambda), \\
\mr{v}_{-}(\lambda) & =\frac{1}{\|T_{\lambda}(-1) w_{-}(C_{-}, \lambda)\|_{\mb{C}^{2}} }
T_{\lambda}(-1) w_{-}(C_{-}, \lambda). 
\end{split}
\end{equation}
These vectors satisfy
\begin{equation}\label{EFA1}
F_{\lambda}(x) \mr{v}_{+}(\lambda) =z_{+}(C_{+}, \lambda)^{x-1} T_{\lambda}(0) \mr{v}_{+}(\lambda),\quad 
F_{\lambda}(x) \mr{v}_{-}(\lambda) =z_{-}(C_{-}, \lambda)^{x+1} T_{\lambda}(-1)^{-1} \mr{v}_{-}(\lambda)
\end{equation}
We set 
\begin{equation}\label{esp1}
S=\{\zeta \in S^{1} \mid |\re(\zeta)| \leq \rho \}. 
\end{equation}
Then, for $\lambda \in \mb{C}^{\times}$, we have 
$|z_{\pm}(C_{\pm},\lambda)|=1$ if and only if $\lambda \in S$. 
This means that the absolute value of all of the eigenvalues of $T_{\lambda}(x)$ is one. 
Therefore, it is not hard to conclude, with Theorem $\ref{TM1}$ that there are no eigenvalues of $U$ on $S$. 
The concrete form of the vectors $\mr{v}_{\pm}(\lambda)$ are given by  
\[
\mr{v}_{+}(\lambda) =\frac{1}{D_{+}(\lambda)} 
\begin{bmatrix}
{\dsp \ol{\alpha_{0}} \left(Z_{+}(\lambda) -\lambda^{-1}\right) }\\[7pt]
{\dsp \ol{\beta_{0}} \left(Z_{+}(\lambda) -\lambda^{-1} \right) 
-\ol{\beta} \lambda  }
\end{bmatrix}, \quad 
\mr{v}_{-}(\lambda) =\frac{1}{D_{-}(\lambda)} 
\begin{bmatrix}
{\dsp \left(
\lambda +\beta_{0} \ol{\beta} \lambda^{-1} \right) Z_{-}(\lambda) -\rho^{2} }\\[7pt]
{\dsp 
- \alpha_{0} \ol{\beta} \lambda^{-1} Z_{-}(\lambda)
 } 
\end{bmatrix}, 
\]
where $D_{\pm}(\lambda)$ are normalization terms and $Z_{\pm}(\lambda)$ are given by 
\[
Z_{+}(\lambda)=\alpha_{+} z_{+}(C_{+},\lambda),\quad Z_{-}(\lambda) =\alpha_{-} z_{-}(C_{-},\lambda). 
\] 
By the assumption $\eqref{AssA1}$, $z=Z_{\pm}(\lambda)$ are the solutions to the equation 
\[
z^{2}-2J(\lambda) z +\rho^{2}=0. 
\]
The denominators $\ispa{\mr{v}_{+}(\lambda), \mr{v}_{-}(\lambda)^{\perp}}_{\mb{C}^{2}}$ and 
$\ispa{\mr{v}_{-}(\lambda), \mr{v}_{+}(\lambda)^{\perp}}_{\mb{C}^{2}}$ in $\eqref{VX000}$ do not vanish if $|\lambda| \neq 1$ 
because when $\lambda \in \mb{C}^{\times} \setminus S^{1}$ the two unit vectors $v_{\pm}(\lambda)$ are linearly independent. 
We set 
\[
s=\re(\beta_{0}\ol{\beta}),\quad t=\im (\beta_{0} \ol{\beta}). 
\]
Then the denominators are given by 
\[
\ispa{\mr{v}_{+}(\lambda), \mr{v}_{-}(\lambda)^{\perp}}_{\mb{C}^{2}} 
=-\ispa{\mr{v}_{-}(\lambda), \mr{v}_{+}(\lambda)^{\perp}}_{\mb{C}^{2}} 
=\frac{2\ol{\beta} }{D_{+}(\lambda)D_{-}(\lambda)} 
Z_{-}(\lambda)\left(
1-s -J(\lambda) Z_{-}(\lambda)
\right). 
\]
Thus we have to consider the zeros of 
\begin{equation}\label{DN1}
f(\zeta)=1-s -J(\zeta) Z_{-}(\zeta_{(\pm)})
\end{equation}
as a function in $\zeta \in S^{1}$, where, in what follows, we sometimes write 
\[
F(\zeta_{(-)})=\lim_{r \uparrow 1}f(r\zeta),\quad 
F(\zeta_{(+)})=\lim_{r \uparrow 1}f(r^{-1}\zeta) \quad (\zeta \in S^{1}) 
\]
for a function $F$ on $\mb{C} \setminus S^{1}$, if the limit exists. 
We set $\zeta=z+iy \in S^{1}$. For $\zeta \in S$, we see $Z_{-}=x+iw$ with $w=\pm \sqrt{\rho^{2}-x^{2}}$ and 
thus $f(\zeta)=1-s-x^{2}-ixw$ which is zero if and only if $x=\pm \rho$ and $\rho^{2}=1-s$. 
When $\rho^{2}=1-s$, the singularity coming from $f(\zeta)$ is $1/\sqrt{\rho^{2}-x^{2}}$ as will be seen below. 
But this is integrable on the interval $[-\rho,\rho]$. Therefore, we take the pointwise limit 
of $\re \mr{x}(r\zeta)$ to compute the matrix-valued measure $d\Sigma(\zeta)|_{S}$ restricted to $S$. 
When $\zeta \not\in S$, we see $Z_{-}=x +w$ with $w=\pm \sqrt{x^{2}-\rho^{2}}$. 
In this case $f(\zeta)$ vanishes if and only if $\zeta$ is the following four points, 
\begin{equation}\label{EVS}
\zeta_{*}=e^{i\theta_{*}}:=x_{*} +i\sqrt{1-x_{*}^{2}},\quad -\ol{\zeta_{*}},\quad 
-\zeta_{*},\quad \ol{\zeta_{*}},  
\end{equation}
where a positive real number $x_{*}$ is given by 
\[
x_{*}=\frac{1-s}{\sqrt{2-2s -\rho^{2}}}. 
\]
It is well-known (see \cite{KS}) that these points are actually eigenvalues of $U$, 
which will also be shown as a point mass of $\Sigma$ later. 
We have 
\[
\begin{split}
\ispa{\mr{z}_{L}(0)e_{L},\mr{v}_{+}(\lambda)^{\perp}}_{\mb{C}^{2}} 
& =\frac{\ol{\beta} \lambda}{D_{+}},\\ 
\ispa{\mr{z}_{R}(0)e_{R},\mr{v}_{-}(\lambda)^{\perp}}_{\mb{C}^{2}} 
& =\frac{1}{D_{-}} (\lambda Z_{-} -\rho^{2}) =
\frac{Z_{-}}{D_{-}} (Z_{-}-\lambda^{-1}) 
\end{split}
\]
These formulas lead us to obtain 
\begin{equation}\label{X0C}
\begin{split}
\lambda \mr{x}_{0}(\lambda) e_{L} & =
\frac{1}{2 [1-s-J(\lambda) Z_{-}(\lambda)]}
\begin{bmatrix}
\beta_{0} \ol{\beta} -1 + \lambda Z_{-}(\lambda) \\[7pt]
-\alpha_{0} \ol{\beta} 
\end{bmatrix} \\[7pt]
\lambda \mr{x}_{0}(\lambda) e_{R} & = 
\frac{1}{2 [1-s-J(\lambda)Z_{-}(\lambda)]} 
\begin{bmatrix}
\ol{\alpha_{0}} \beta  \\[7pt]
\ol{\beta_{0}} \beta -1 + \lambda Z_{-}(\lambda) 
\end{bmatrix}
\end{split}
\end{equation}
Then we see 
\[
\begin{split}
\zeta \mr{x}_{0}(\zeta)e_{L} & =
\frac{1}{2[1-s-x^{2}-ixw]} 
\begin{bmatrix}
s-y^{2}-yw +i(xy+xw+t) \\
-\alpha_{0}\ol{\beta}
\end{bmatrix} \\[7pt]
\zeta \mr{x}_{0}(\zeta)e_{R} & = 
\frac{1}{2 [1-s-x^{2} -ixw] } 
\begin{bmatrix}
\ol{\alpha_{0}} \beta \\
s-y^{2}-yw +i(xy+xw-t) 
\end{bmatrix}
\end{split}
\]
These formulas make us to obtain the following. 
\begin{equation}\label{2PS}
d\Sigma(\zeta)|_{S}=
\frac{-w}{(1-s)^{2}-[2(1-s) -\rho^{2}]x^{2}}
\begin{bmatrix}
[(1-s)y+tx]  & -i \ol{\alpha_{0}} \beta x \\
i\alpha_{0}\ol{\beta} x & [(1-s)y-tx]
\end{bmatrix}
\end{equation}
Next we consider $\Sigma$ on $S^{1} \setminus S$. On the $S^{1} \setminus S$, we have 
\[
\lim_{r \uparrow 1} (Z_{-}(r\zeta) -Z_{-}(r^{-1} \zeta)) =0 \quad (\zeta \in S^{1} \setminus S \cup \{\pm \zeta_{*}, \pm \ol{\zeta_{*}} \}). 
\]
On the domains 
\[
\Omega_{+}:=\{\lambda \in \mb{C} \mid \re(\lambda) >0,\ \lambda \not\in S \cup i\mb{R}\},  \quad 
\Omega_{-}:=\{\lambda \in \mb{C} \mid \re(\lambda) <0,\ \lambda \not\in S \cup i\mb{R}\}, 
\]
the function $Z_{-}(\lambda)$ is then given as 
\[
Z_{-}(\lambda)=
\begin{cases}
J(\lambda) +\sqrt{J(\lambda)^{2}-\rho^{2}} & \lambda \in \Omega_{+}, \\
J(\lambda) -\sqrt{J(\lambda)^{2} -\rho^{2}} & \lambda \in \Omega_{-}. 
\end{cases}
\]
We denote $\zeta_{o}=x_{o}+iy_{o}$ one of the eigenvalues $\{\pm \zeta_{*},\pm \ol{\zeta_{*}}\}$ 
of $U$ and write $e^{i\theta_{o}}=\zeta_{o}$. For any $r \in (1/2,1)$ and small $\ep>0$, we define 
\[
J_{r,\ep}=\{t e^{i\theta} \mid r \leq t \leq r^{-1},\,|\theta -\theta_{o}| \leq \ep\}, 
\quad A_{\ep}=J_{r,\ep} \cap S^{1}=\{e^{i\theta} \mid |\theta -\theta_{o}| \leq \ep \}, 
\] 
where $\ep>0$ is chosen so that $J_{r,\ep}$ is contained in one of $\Omega_{\pm}$ which contains $\zeta_{o}$ 
and $J_{r,\ep}$ contains only $\zeta_{o}$ among $\{\pm \zeta_{*},\pm \ol{\zeta_{*}}\}$. 
For simplicity we set 
\[
f(\lambda)=1-s-J(\lambda)Z_{-}(\lambda). 
\]
Then $\zeta_{o}$ is a zero of $f$ of order $1$ and there are no zeros of $f$ on a neighborhood of $J_{r,\ep}$. 
Thus, for any holomorphic function $h(\lambda)$ near $J_{r,\ep}$, Cauchy's formula gives 
\begin{equation}\label{CF1}
\frac{1}{2\pi i} \int_{\partial J_{r,\ep}}  \frac{h(\lambda)}{f(\lambda)}\,\frac{d\lambda}{\lambda}
= \frac{h(\zeta_{o})}{\zeta_{o} f'(\zeta_{o})}
=- \frac{h(\zeta_{o})} {\zeta_{o} Z_{-}(\zeta_{o}) Z_{-}'(\zeta_{o})}, 
\end{equation}
where $\partial J_{r,\ep}$ is the boundary of $J_{r,\ep}$ with the counterclockwise direction. 
By computing the contour integral in the right-hand side of $\eqref{CF1}$, we see 
\begin{equation}\label{WL1}
\lim_{r \uparrow 1} 
\left[
\int_{A_{\ep}} \frac{h(r\zeta)}{f(r\zeta)} \,\dbar \zeta -
\int_{A_{\ep}} \frac{h(r^{-1}\zeta)}{f(r^{-1}\zeta)} \,\dbar \zeta
\right]=
\frac{h(\zeta_{o})} {\zeta_{o} Z_{-}(\zeta_{o}) Z_{-}'(\zeta_{o})}.
\end{equation}
The set of all Laurent polynomials is dense in the space of continuous functions on $A_{\ep}$. 
Thus, the formula $\eqref{WL1}$ still holds for any continuous function $h$ on $A_{\ep}$. 
From $\eqref{WL1}$ we obtain, for $\re(\zeta_{o})>0$, 
\begin{equation}\label{WL2}
\begin{split}
\Sigma(\{\zeta_{o}\}) & =
\frac{1}{2\zeta_{o} Z_{-}(\zeta_{o}) Z_{-}'(\zeta_{o})}
\begin{bmatrix}
[\beta_{o} \ol{\beta} -1 +\zeta_{o} Z_{-}(\zeta_{o})] \delta_{\zeta_{o}} & \ol{\alpha_{0}} \beta \delta_{\zeta_{o}} \\
-\alpha_{0} \ol{\beta} \delta_{\zeta_{o}} & [\ol{\beta_{o}} \beta -1 +\zeta_{o} Z_{-}(\zeta_{o})] \delta_{\zeta_{o}}
\end{bmatrix} \\
& = 
\frac{x_{o}\sqrt{x_{o}^{2}-\rho^{2}}}{2y_{o}(1-s)^{2}} 
\begin{bmatrix}
[(1-s) y_{o}+tx_{o}] \delta_{\zeta_{o}} & -i\ol{\alpha_{0}} \beta x_{o} \delta_{\zeta_{o}} \\
i \alpha_{0} \ol{\beta} x_{o} & [(1-s) y_{o}-tx_{o}] \delta_{\zeta_{o}}
\end{bmatrix}
\end{split} 
\end{equation}
where $\delta_{\zeta_{o}}$ is the delta measure at the point $\zeta_{o}$, and for $\re(\zeta_{o})<0$, 
\begin{equation}\label{WL3}
\Sigma(\{\zeta_{o}\})  
 = 
-\frac{x_{o}\sqrt{x_{o}^{2}-\rho^{2}}}{2y_{o}(1-s)^{2}} 
\begin{bmatrix}
[(1-s) y_{o}+tx_{o}] \delta_{\zeta_{o}} & -i\ol{\alpha_{0}} \beta x_{o} \delta_{\zeta_{o}} \\
i \alpha_{0} \ol{\beta} x_{o} & [(1-s) y_{o}-tx_{o}] \delta_{\zeta_{o}}
\end{bmatrix}.
\end{equation}
From this we conclude the following. 
\begin{thm}\label{TP1}
The positive-matrix-valued measure $\Sigma$ for the two-phase model $U(\ccal)$ with 
the assumption $\eqref{AssA1}$ is given by the following. 
\[
\begin{split}
d\Sigma(\zeta) & =
\frac{\sqrt{\rho^{2}-x^{2}}} {(1-s)^{2}-(2-2s -\rho^{2}) x^{2}} 
\begin{bmatrix}
(1-s)|y| + \sgn(y)tx & -\sgn(y) i \ol{\alpha_{0}} \beta x \\
i \sgn(y) \ol{\alpha_{0}} \beta x & (1-s) |y| - \sgn(y) tx
\end{bmatrix} \chi_{S} \dbar \zeta \\
& \hspace{10pt} + \frac{x_{*}\sqrt{x_{*}^{2}-\rho^{2}}}{2y_{*}(1-s)^{2}} 
\begin{bmatrix}
[(1-s) y_{*}+tx_{*}]  & -i\ol{\alpha_{0}} \beta x_{*}  \\
i \alpha_{0} \ol{\beta} x_{*} & [(1-s) y_{*}-tx_{*}] 
\end{bmatrix} \left( \delta_{\zeta_{*}} +\delta_{-\zeta_{*}} \right) \\
& \hspace{10pt} + \frac{x_{*}\sqrt{x_{*}^{2}-\rho^{2}}}{2y_{*}(1-s)^{2}} 
\begin{bmatrix}
[(1-s) y_{*}-tx_{*}]  & i\ol{\alpha_{0}} \beta x_{*}  \\
-i \alpha_{0} \ol{\beta} x_{*} & [(1-s) y_{*}+tx_{*}] 
\end{bmatrix} \left( \delta_{\ol{\zeta_{*}}} + \delta_{-\ol{\zeta_{*}}} \right), 
\end{split}
\]
where we set $\zeta=x+iy \in S^{1}$ with $x,y \in \mb{R}$, 
\[
\sgn(y)=
\begin{cases}
+1 & (y>0) \\
-1 & (y<0)
\end{cases}
\]
and $\zeta_{*}=x^{*}+iy_{*}$ is given as 
\[
x_{*}=\frac{1-s}{\sqrt{2-2s -\rho^{2}}},\quad y_{*}=\sqrt{1-x_{*}^{2}} =\sqrt{\frac{1-\rho^{2}-s^{2}}{2-2s -\rho^{2}}}. 
\]
\end{thm}
\par
\medskip
\appendix
\section{Proof of Theorem $\ref{TM1}$}
Theorem $\ref{TM1}$ is proved in \cite{KKK}, \cite{KS}, \cite{MSSSS}. 
But the way of presentation here is somehow different from them. 
Thus, we give its proof for completeness. 
Suppose that $f \in \map(\mb{Z},\mb{C}^{2})$ satisfies $U(\ccal)f=\lambda f$ with $\lambda \in S^{1}$. 
Then by definition of $U(\ccal)$, we have 
\begin{equation}\label{EE1}
\pi_{L}\ccal(x+1) f(x+1)=\lambda \pi_{L}f(x),\quad 
\pi_{R}\ccal(x-1) f(x-1)=\lambda \pi_{R}f(x)
\end{equation}
Shifting the variable $x$, we get 
\begin{equation}\label{EE2}
\pi_{L}\ccal(x)f(x)=\lambda \pi_{L} f(x-1),\quad \pi_{R} \ccal(x) f(x) = \lambda \pi_{R} f(x+1). 
\end{equation}
The second equation of $\eqref{EE2}$ gives 
\begin{equation}\label{EE3}
\pi_{R}f (x+1) = \lambda^{-1} \pi_{R} \ccal(x) f(x). 
\end{equation}
The first equation of $\eqref{EE1}$ gives
\[
a_{x+1}\pi_{L} f(x+1)+\pi_{L} \ccal(x+1) \pi_{R} f(x+1) =\lambda \pi_{L}f (x). 
\]
Substituting $\eqref{EE3}$ into the above and dividing both side by $a_{x+1}$ show 
\begin{equation}\label{EE4}
\pi_{L}f(x+1) =\frac{1}{a_{x+1}} 
\left(
\lambda \pi_{L} -\lambda^{-1} \pi_{L} \ccal(x+1) \pi_{R} \ccal(x) 
\right) f(x)
\end{equation}
Summing $\eqref{EE3}$ and $\eqref{EE4}$ then gives 
\begin{equation}\label{EE5}
f(x+1)=T_{\lambda}(x)f(x) \quad (x \in \mb{Z}),  
\end{equation}
where $T_{\lambda}(x)$ is given in $\eqref{TRM1}$. From $\eqref{EE5}$, 
it is easy to show that $f(x)=F_{\lambda}(x)f(0)$. 
Conversely, a direct computation using Lemma $\ref{FMS}$ shows that 
the function defined in $\eqref{EF1}$ is an eigenfunction of $U(\ccal)$ with the eigenvalue $\lambda$. 
Therefore $\mcal^{\lambda}=\Phi_{\lambda}(\mb{C}^{2})$. 
Since $\Phi_{\lambda} \colon \mb{C}^{2} \to \map(\mb{Z},\mb{C}^{2})$ is obviously injective, 
we have $\dim \mcal^{\lambda}=2$. Next, we introduce a map 
\[
J \colon \map(\mb{Z},\mb{C}^{2}) \to \map (\mb{Z},\mb{C}^{2}) 
\]
defined as 
\[
(Jf)(x)=
\begin{bmatrix}
f_{L}(x-1) \\
f_{R}(x)
\end{bmatrix},\quad \text{where} \quad 
f(x)=
\begin{bmatrix}
f_{L}(x) \\
f_{R}(x)
\end{bmatrix}.
\]
Let $f=\Phi_{\lambda}(w)$ with $w \in \mb{C}^{2}$. Using $\eqref{EE2}$ we see 
\[
a_{x}\pi_{L} f(x) =\lambda \pi_{L}f(x-1) -\pi_{L} \ccal(x) \pi_{R} f(x), \quad 
\pi_{R} \ccal(x) \pi_{L} f(x) -\lambda \pi_{R} f(x+1) =-d_{x} \pi_{R} f(x). 
\]
This is equivalent to 
\begin{equation}\label{TM22}
Jf(x+1) = S_{\lambda}(x) Jf(x),\quad  
S_{\lambda}(x) =\frac{1}{a_{x}}
\begin{bmatrix}
\lambda & -b_{x} \\
c_{x} & \lambda^{-1} \triangle_{x}
\end{bmatrix}. 
\end{equation} 
The matrix $S_{\lambda}(x)$ is also called the {\it transfer matrix} (\cite{KS}, \cite{MSSSS}). 
Compared with $T_{\lambda}(x)$ which we used throughout the paper, 
$S_{\lambda}(x)$ has a nice property $\det S_{\lambda}(x)=d_{x}/a_{x}$ 
and hence $|\det S_{\lambda}(x)|=1$ whereas $\det T_{\lambda}(x) =d_{x}/a_{x+1}$ 
whose absolute value is, in general, not the unity. 
Let $v,w \in \mb{C}^{2}$ and suppose, as in \cite{MSSSS}, that $\Phi_{\lambda}(v) \in \ell^{2}(\mb{Z},\mb{C}^{2})$ 
and that $\|\Phi_{\lambda}(w)(x)\|_{\mb{C}^{2}}$ is bounded. 
By the definition of $J$, it is easy to show that $\|(J\Phi_{\lambda}(w))(x)\|_{\mb{C}^{2}}$ is bounded. 
Since $J$ is unitary, $J\Phi_{\lambda}(w) \in \ell^{2}(\mb{Z},\mb{C})$. 
Thus we can take a sequence $\{x_{n}\}$ of positive integers such that $\|(J\Phi_{\lambda}(w))(x_{n})\|_{\mb{C}^{2}} \to 0$ 
as $n \to \infty$. 
But, by $\eqref{TM22}$, we see that the function 
\[
W(x)=|\det [(J\Phi_{\lambda}(v))(x),\ (J\Phi_{\lambda}(w))(x)]|
\]
in $x \in \mb{Z}$ is constant. We have $W(x_{n}) \to 0$ as $n \to \infty$ and hence $W(x)$ is identically zero. 
This shows that $J\Phi_{\lambda}(v)(x)$ and $J\Phi_{\lambda}(w)(x)$ are linearly dependent. 
Let $s_{0}, t_{0} \in \mb{C}$ satisfy
\[
s_{0}J\Phi_{\lambda}(v)(0)+t_{0}J\Phi_{\lambda}(w)(0)=0. 
\]
Then for $x\geq 1$, we see 
\[
s_{0}J\Phi_{\lambda}(v)(x)+t_{0}J\Phi_{\lambda}(w)(x)=S_{\lambda}(x-1) \cdots S_{\lambda}(0) 
\left(
s_{0}J\Phi_{\lambda}(v)(0)+t_{0}J\Phi_{\lambda}(w)(0)
\right)=0.  
\]
Similar computation for $x \leq -1$ shows that $s_{0}J\Phi_{\lambda}(v)(x)+t_{0}J\Phi_{\lambda}(w)(x)=0$ 
for any $x \in \mb{Z}$. Therefore, we conclude $\dim \mcal^{\lambda} \cap \ell^{2}(\mb{Z},\mb{C}^{2}) \leq 1$. 
\hfill$\square$

\end{document}